\documentclass{article}
\usepackage{verbatim}

\usepackage{times}
\usepackage[pdftex]{graphicx}
\usepackage{amssymb,amsfonts,amsmath,amsthm}
\usepackage{float}
\usepackage{xcolor}
\usepackage{enumitem}
\usepackage{bbm}

\newtheorem{theorem}{Theorem}
\newtheorem{proposition}{Proposition}
\newtheorem{lemma}{Lemma}
\newtheorem{ass}{Assumption}
\newtheorem{definition}{Definition}
\newtheorem{remark}{Remark}

\newcommand\EE {\mathbb E}
\newcommand\FF {\mathbb F}
\newcommand\NN {\mathbb N}
\newcommand\RR {\mathbb R}
\newcommand\PP {\mathbb P}
\newcommand\QQ {\mathbb Q}
\newcommand\WW {\mathbb W}

\def\qed{\hskip6pt\vrule height6pt width5pt depth1pt}

\title{Optimal contract design via relaxation: application to the problem of brokerage fee for a client with private signal.}
\author{G. A. Alvarez and S. Nadtochiy\footnote{Partial support by the NSF grant DMS-2205751 is acknowledged by both authors.}}
\begin{document}
\maketitle

\begin{abstract}
In this paper we show how the relaxation techniques can be used to establish the existence of an optimal contract in presence of information asymmetry. The method we illustrate was initially motivated by the problem of designing optimal brokerage fees, but it does apply to other optimal contract problems, in which (i) the agent controls linearly the drift of a diffusion process, (ii) the direct dependence of the principal's and the agent's objectives on the strategy of the agent is of a special form, and (iii) the space of admissible contracts is compact. This method is then applied to establish existence of an optimal brokerage fee in a market model with a private trading signal observed by the broker's client but not by the broker.
\end{abstract}

\section{Introduction}

The problem of optimal contract design (a.k.a. the principal-agent problem) involves two parties, a principal and an agent, who enter into an agreement, a.k.a. the contract, which is designed by the principal. According to this agreement, the principal promises the agent a payment contingent on the agent's action and on the realized values of the observed (random) states. The agent determines the optimal action by maximizing his objective, which depends on his action, on the observed state processes, and on the contract. The principal chooses an optimal contract as to maximize her objective, which depends on the agent's optimal action (given the contract), as well as on the states and on the contract itself. If the agent's optimization is performed in continuous time, and if the observations and states are given by stochastic processes, the optimal contract problem results in a pair of coupled stochastic control problems. The formal framework for such problems is developed, for example, in \cite{HolmstromMilgrom}, \cite{CvitanicZhang}, \cite{CadenillasCvitanic}, \cite{Sannikov}, \cite{TouziPossamaiCvitanic}. Note, however, that the latter papers do not directly address the question of existence of an optimal contract in a general framework, focusing instead on the solutions to specific problems and on reducing the optimal contract design to more conventional stochastic control problems. However, in many cases of interest (discussed below), such a reduction does not immediately yield the existence of an optimal contract.  

On contrary, \cite{Kadan} addresses the question of existence of an optimal contract. Most of the latter paper is devoted to the relaxed formulation of the optimal contract problem, in which the contract is allowed to depend on the additional (artificially generated) random state. Then, under the assumptions of compactness of the agent's actions and of the continuity of objectives, \cite[Theorem 5.11]{Kadan} shows the existence of an optimal contract in the relaxed formulation. If, in addition, the space of admissible contracts is compact, \cite[Proposition 12.5]{Kadan} shows the existence of an optimal contract in the strong formulation. In the present paper, the compactness of the set of admissible contracts is assumed a priori, but the set of admissible actions of the agent, in the strong formulation of the problem, is not compact (as is the case in most relevant stochastic control problems). To address the latter challenge, we consider a relaxation of the control problem of the agent and show (under appropriate structural assumptions on the controlled state and on the objectives) that the set of admissible actions in the relaxed problem is compact, while the objectives remain continuous. This yields existence of an optimal contract in the relaxed problem (Theorem \ref{existence:weak:opt}). Then, we show (under appropriate convexity assumptions) that an optimal contract in the relaxed problem is also optimal in the strong formulation (Theorem \ref{cor:main}). 

\smallskip

It is worth explaining why we resort to the abstract existence results, as opposed to using the powerful methods developed recently, for example, in \cite{Sannikov}, \cite{TouziPossamaiCvitanic}. The latter approach allows one to reduce an optimal contract problem to the more standard problem of a controlled diffusion. However, the resulting control problem involves the control of diffusion coefficient (as opposed to controlling the drift only), for which there are no general existence results. The situation becomes even more complicated if one incorporates additional information asymmetry into the model. Namely, in some models, it is natural to assume that certain coordinates of the state process are observed by the agent but not by the principal. This reflect the fact that the agent has private information, and a specific example of such a setting is described in the subsequent paragraph. The presence of private signal of the agent introduces the new form of information asymmetry, which is different from the standard second- and third-best settings.\footnote{This information asymmetry may seem similar to the one appearing in the third-best setting. Indeed, one may view the private signal of the agent as the agent's ``type". However, unlike the actual agent's type in the third-best setting, the private signal is only observed by the agent \emph{after} the contract is determined, which makes it impossible to construct the ``menu of contracts" typically used in the third-best case.} Applying the methods of \cite{Sannikov}, \cite{TouziPossamaiCvitanic} to a problem with such information asymmetry, one reduces the optimal contract problem to a problem of controlled diffusion but with the additional informational constraint: at the terminal time, certain coordinates of the controlled state process must be measurable w.r.t. a smaller sigma algebra (than the one given by the terminal value of the filtration to which the controls are to be adapted). To the best of our knowledge, to date, there exist no tractable characterizations of solutions to such problems, and it is not clear how to establish their existence.\footnote{See \cite{Williams} for the solution in a concrete example.} We leave the latter challenges for future research.

The financial problem that motivated our analysis is the optimal design of brokerage fees. Consider an agent who trades a financial asset via a broker: the agent decides on the trading rate, while the broker (i.e. the principal) decides what fee the agent has to pay for trading. The latter fee is allowed to depend on the state processes observed by the broker (e.g., on the price process). However, the agent may have a private trading signal (e.g., obtained via the use of proprietary predictive factors) which is not observed by the broker (as the later is not in the business of designing such predictive factors). A simplified version of this model, with symmetric observations, was considered recently in \cite{AlonsoWebsterNadtochiy}. However, the presence of non-standard information asymmetry makes the problem much more challenging. Intuitively, the broker wants to incentivize the agent to use his private signal in order to make him more profitable and, therefore, more agreeable to pay higher trading fees. However, as the trading signal is not observed by the broker, the exact structure of an optimal trading fee is not clear. In this paper, we use the aforementioned relaxation method to prove the existence of an optimal fee, leaving its characterization for future research.

\smallskip

The remainder of the paper is organized as follows. In Section \ref{strong:opt:contract} we introduce the strong formulation of the optimal contract problem where (i) the agent controls linearly the drift of a diffusion process, (ii) the direct dependence of the principal's and the agent's objectives on the strategy of the agent is of a special form, and (iii) the space of admissible contracts is compact.
In Section \ref{relaxed:optimal:contract}, we introduce a relaxed version of the optimal contract problem and prove the existence of an optimal contract in such a relaxed problem by showing that any epsilon-optimal sequence of contracts has a limit point that is optimal.
In Section \ref{solution_strong_problem} we establish a precise connection between the relaxed and the strong formulations of the optimal contract problem, and we show that, under an appropriate concavity assumption, every optimal contract in the relaxed problem is a solution to the strong problem.
Finally, in Section \ref{opt:brok:as:info}, we apply our results to the problem of designing optimal brokerage fee for a client with a private trading signal, where we prove that any epsilon-optimal sequence of contracts has a limit point that is optimal.

\section{Strong formulation of the optimal contract problem}
\label{strong:opt:contract}

We start by introducing the optimal contract problem we study in this paper. We fix $T < \infty$, $k \geq 1$ and consider $\Omega = C([0,T],\RR^k)$, the canonical space equipped with the Wiener measure $\WW$.\footnote{One may choose a different reference measure, instead of $\WW$. For example, the results that follow also hold if $\WW$ is replaced by a scaled Wiener masure: i.e., the distribution of a standard Brownian motion whose components are multiplied by positive constants.} We denote by $X$ the canonical random element in $\Omega$ and by $\FF$ its completed (with respect to $\WW$) natural filtration.
Next, we introduce the controls of the agent as follows. Let $l\geq 1$, we consider the progressively measurable mappings $A$, $b$: 
\begin{equation*}
    \begin{split}
        A&:  \Omega \mapsto C([0,T],\RR^{k\times l}), \\
        b&:  \Omega \mapsto C([0,T],\RR^{l}),
    \end{split}
\end{equation*}
that satisfy
\begin{enumerate}
    \item $A_t(\cdot), b_t(\cdot)$ are continuous with respect to $\|.\|_{C}$, for any $0 \leq t\leq T$.
    \item $\sup_{t\in [0,T]}\sup_{x\in \Omega}|b_{t\wedge\tau_n}(x)| + \sup_{t\in [0,T]}\sup_{x\in \Omega}|A_{t\wedge \tau_n}(x)| < \infty,$ 
\end{enumerate}
where $\|.\|_{C}$ denotes the supremum norm on $\Omega$ and 
\begin{align}
\tau_n :=\inf\left\{t\geq 0: |X_t|\geq n\right\},\quad n\geq 1.\label{eq.taun.def}
\end{align}

\begin{definition}\label{def:A}
Agent's admissible control is any $\RR^{k}$-valued locally square integrable process $\nu$, constructed on $(\Omega,\mathbb{F},\mathbb{W})$, such that:
\begin{itemize}
\item[a)] $\bar M_T(\nu) := \exp\left( -\frac{1}{2}\int_0^T\|\nu_s\|^2\,ds + \int_0^T \nu^\top_s\,dX_s\right)$ 
satisfies $$\EE^{\mathbb{W}}\left[\log(\bar{M}_T(\nu))\bar M_T(\nu)\right]< \infty,$$
implying $\EE^{\WW}\bar{M}_T(\nu)= 1$ (via Jensen's inequality and Vall\'ee-Poussin theorem),
\item[b)] $\mathbb{W}$-a.s., the inequality $b_t + A_t\,\nu_t \leq 0$ holds for all $t\in[0,T]$.
\end{itemize}
We denote the set of all admissible controls of the agent by $\mathcal{A}$.
\end{definition}

Girsanov's theorem (\cite[Theorem 5.1]{KaratzasShreve}) implies that the state process corresponding to the agent's control $\nu$ satisfies the following diffusion equation: 
\begin{equation}\label{eqn.state}
    dX_t = \nu_tdt+dB^\nu_t,
    \quad X_0 = x_0,
\end{equation}
where $B^\nu$ is a $k$-dimensional standard Brownian motion under a probability measure $\bar\QQ^\nu$ on $(\Omega,\mathcal{F}_T)$, defined via
\begin{align*}
& d\bar\QQ^\nu/d\mathbb{W} = \bar{M}_T(\nu).
\end{align*} 


The above construction explains the first condition in Definition \ref{def:A}: the agent controls the drift of the state process by choosing an associated measure with the Radon-Nikodym derivative $\bar M(\nu)$.
To explain the second condition of Definition \ref{def:A}, assume that $\nu$ is an affine function of a progressively measurable process $\pi$, so that $\nu_t= \tilde b_t + \tilde A_t \pi_t$, with progressively measurable $(\tilde b,\tilde A)$ and with $\pi$ that can be chosen freely (subject to integrability conditions) by the agent. Then, the range of resulting processes $\nu$ will be given by a linear space, as in the second condition of Definition \ref{def:A}. Thus, the latter definition states that the agents controls linearly the drift of the state process (this is shown in more detail in Section \ref{opt:brok:as:info}).

\begin{remark}
The pair $(\nu,\bar\QQ^\nu)$ is known as weak control and it is commonly used in the literature on optimal contract with moral hazard (see, e.g., \cite{Sannikov}, \cite{TouziPossamaiCvitanic}). We emphasize its difference with the agent's relaxed controls introduced later in this paper.
\end{remark}

Next, we introduce the set of admissible contracts $\xi$ that the principal can proposes to the agent.
\begin{definition}\label{def:P}
The set of admissible contracts $\mathcal{C}$ is a collection of continuous mappings $\xi$ from ${\Omega}$, equipped with the uniform norm, to $\RR$, such that: $\mathcal{C}$ is equicontinuous on any compact set, pointwise bounded, and closed with respect to the topology of uniform convergence on compact sets.
\end{definition}

\begin{remark}
In the above definition, the mappings $\xi$ can be restricted to depend only on certain components of the canonical process $X$. This feature allows us to include the information asymmetry discussed in the introduction. We illustrate this in more detail in Section \ref{opt:brok:as:info}.  
\end{remark}

\smallskip

For any admissible contract $\xi\in\mathcal{C}$, the agent's value and objective are given by:  
\begin{align}
&\bar{V}_a(\xi):=\sup_{\nu\in\mathcal{A}} 
\bar{J}_a\left(\nu,\xi\right), \label{eq.agentProblem}\\
&\bar{J}_a(\nu,\xi):=\EE^{\WW} U_a\left(\xi(X),X,\bar M_T(\nu)\right),\nonumber
\end{align}
where $U_a:\RR\times \Omega \times \RR \mapsto \RR$ is the utility of the agent (which is a Borel-mesurable function).

\smallskip

It is worth commenting on why and how the agent's utility depends on $\bar M_T(\nu)$. In case the agent is only affected by the state $X$ and by the contract's payment $\xi(X)$, via a function $\bar U_a(\xi,X)$, it is natural to choose his objective as 
\begin{align*}
&\bar{J}_a(\nu,\xi)=\EE^{\bar\QQ^\nu} \bar U_a(\xi(X),X) = \EE^{\WW} \left[U_a(\xi(X),X) \bar M_T(\nu)\right],
\end{align*}
where we recover $U_a(\xi,x,m)=\bar U_a(\xi,x)\,m$. However, the main reason to introduce the dependence on $\bar M_T(\nu)$ in the utility function $U_a$ is to include the cases where agent is affected by $\nu$ directly. Assume, for example, that the agent is affected by $X$ and $\xi(X)$, but is also a subject to quadratic penalty for large $\nu$. Then, it is natural to define
\begin{align*}
& \bar U_a(\xi,X,\nu) := \tilde U(\xi,X) - \phi\int_0^T \|\nu_s\|^2 ds,
\quad \bar{J}_a(\nu,\xi):=\EE^{\bar\QQ^\nu} \bar U_a(\xi(X),X,\nu),
\end{align*}
with some auxiliary function $U_1$.
Assuming sufficient integrability of $\nu$ (which, e.g., can be enforced via the choice of $A,b$), we obtain
\begin{align*}
&\bar{J}_a(\nu,\xi) = \EE^{\WW} \left[\tilde U(\xi(X),X) \bar M_T(\nu)\right] - 2\phi \,\EE^{\bar\QQ^\nu}\log \bar M_T(\nu)\\
&= \EE^{\WW} \left[\tilde U(\xi(X),X) \bar M_T(\nu) - 2\phi \,\bar M_T(\nu)\log \bar M_T(\nu)\right]
=: \EE^{\WW} U_a\left(\xi(X),X,\bar M_T(\nu)\right).
\end{align*}
Another example arises if the penalty for large $\nu$ is multiplicative:
\begin{align*}
& \bar U_a(\xi,X,\nu) := \tilde U(\xi,X)\,e^{-\phi\int_0^T \|\nu_s\|^2 ds},
\quad \bar{J}_a(\nu,\xi):=\EE^{\bar\QQ^\nu} \bar U_a(\xi(X),X,\nu),\\
&\bar{J}_a(\nu,\xi) = \EE^{\WW} \left[\tilde U(\xi(X),X)\,\exp\left(-\frac{1+2\phi}{2}\int_0^T \|\nu_s\|^2 ds + \int_0^T \nu^\top_s\,dX_s\right)\right]\\
& = \EE^{\WW} \left[\tilde U(\xi(X),X)\,\bar{M}^q_T(\nu/q) \right] =: \EE^{\WW} U_a\left(\xi(X),X,\bar M_T(\nu/q)\right),
\end{align*}
where $q:=1/(1+2\phi)$. Of course, there may exist other cases where the dependence on $\nu$ can be represented as the dependence on $\bar M_T(\nu)$.

\smallskip

Next, we introduce the principal's problem. First, we define the set of agent's optimal responses for a given $\xi \in \mathcal{C}$,
\begin{equation*}
    \mathcal{A}^*(\xi) :=\{\nu \in \mathcal{A} : \bar{J}_a\left(\nu,\xi\right) = \bar{V}_a(\xi)\in\RR\},
\end{equation*}
and the restricted set of admissible contracts
\begin{equation*}
     \bar{\mathcal{C}}^{a} := \{\xi\in\mathcal{C}:\, \bar{V}_a(\xi)\geq R_a,\, \mathcal{A}^*(\xi) \neq \emptyset\},
\end{equation*}
which represents the contracts that meet the agent's reservation value $R_a$.
Finally, the principal's value and objective are given by:  
\begin{align}
\bar{V}_p&:=\sup_{\xi\in\bar{\mathcal{C}}^{a}} \bar{J}_p(\xi),\label{eq.principalProblem.def}\\
\bar{J}_p(\xi)&:=\sup_{\nu\in\mathcal{A}^*(\xi)}\EE^{\WW} U_p\left(\xi(X),X,\bar M_T(\nu)\right),\nonumber   
\end{align}
where $U_p:\RR\times \Omega \times \RR \mapsto \RR$ is the utility of the principal (which is a Borel-mesurable function). The reason why $U_p$ depends on $\bar M_T(\nu)$ is the same as the one given above for the agent's utility.

\smallskip

We assume that $\mathcal{A}$ and the utility functions $U_{a/p}$ are chosen so that the associated expectations are well defined for all $\nu\in\mathcal{A}$ and all $\xi\in\mathcal{C}$. This is ensured by the assumptions that follow.




\section{Relaxation of the problem}
\label{relaxed:optimal:contract}

In this section we relax the agent's strong control $\nu$ in a convenient way. The goal of this relaxation is to make the set of (relaxed) controls pre-compact (note that $\mathcal{A}$ is not pre-compact in any conventional topology), while preserving the continuity of the objectives w.r.t. controls.
Notice that any admissible control $\nu$ of the agent can be identified with a probability measure $\bar\QQ^\nu$ on $\Omega$. In some cases, the latter observation may yield the desired relaxation, in which we replace $\nu$ by $\bar\QQ^\nu$. More precisely, such a relaxation works well if the utilities $U_{a/p}\left(\xi,X,\bar M_T(\nu)\right)$ are linear in $\bar M_T(\nu)$, in which case we can represent the objectives as expectations of functions of $(\xi(X),X)$ under $\bar\QQ^\nu$, thus, removing $\bar M_T(\nu)$ and any other source of dependence on $\nu$. However, if the dependence on $\bar M_T(\nu)$ is of a different form (e.g., if the agent or the principal is subjected to quadratic penalty for large $\nu$, as discussed in the previous section), such a trivial relaxation does not remove $\bar M_T(\nu)$, and hence the direct dependence on $\nu$, from the objectives. Thus, we define a relaxed control of the agent as a probability measure on the extended state space $ \Omega \times \RR$, which represents the pairs $(X,\bar M_T)$, but with the second component no longer being a function of the first one.

\smallskip

On the extended state space as $\Omega \times \RR = {C}([0,T],\RR^{k})\times \RR$ we consider the $\RR^l$-valued stochastic process $Z$ (initially constructed on $(\Omega,\mathbb{F})$) defined by
\begin{equation*}
Z_t := \int_0^t(b_s\,ds + A_s\,dX_s),
\end{equation*}
where the processes $A, b$ are defined in the previous section. Notice that the process $Z$ is well defined due to the properties of $A$ and $b$. 
Next, we introduce the agent's relaxed controls as follows.

\begin{definition}\label{def:U}
For any given $A$, $b$ and $Z$, as above, we define $\mathcal{U}$ as the set of all probability measures $\QQ$ defined on $(\Omega\times\RR,\mathcal{F}_T\otimes\mathcal{B}(\RR))$, such that
\begin{itemize}
\item[1.] $\EE^\QQ M = 1$, 
\item[2.] $\QQ(M > 0) = 1$,
\item[3.] $\QQ(dx,\RR)=\mathbb{W}(dx)$,
\item[4.] $\EE^\QQ\left[M\log(M) \right] < \infty$,
\item[5.] $\EE^\QQ\left[M\eta^{\top} \left(Z_{t\wedge \tau_N}-Z_{s\wedge\tau_N} \right)\right] \leq 0$,
for all $0\leq s\leq t\leq T$, $N\geq1$, and all measurable $\RR^l_+$-valued bounded functions $\eta$ of the paths of $X$, such that $\eta(x) = \eta\left(x_{r},r \leq s\wedge\tau_N(x)\right)$ and such that $\eta$ is continuous at $\mathbb{W}$-a.e. path of $X$. In the above, $M$ denotes the last coordinate of the canonical element on $\Omega \times \RR$, and the stopping time $\tau_N$ is defined in \eqref{eq.taun.def}. 
\end{itemize}
\end{definition}

\smallskip

A strong control $\nu$ of the agent corresponds to the relaxed control
\begin{align}
\hat\QQ^\nu(dx,dm):=\mathbb{W}(dx)\,\delta_{\bar M(\nu)}(dm).\label{eq.hatQ.def}
\end{align}
Recall $\bar M(\nu)$ is a function of $X$ that is a Radon-Nikodym derivative of $\QQ^\nu$ w.r.t. $\WW$. The first three conditions of Definition \ref{def:U} reflect this observation. The last condition Definition \ref{def:U} represents the property (b) of Definition \ref{def:A}: this connection is made precise in the proofs of Lemmas \ref{le:strong.sol.le1} and \ref{lemma:strong:is:weak}.

\smallskip

We consider the following relaxation of the agent's control problem:
\begin{align}
&V_a(\xi):=\sup_{\QQ\in\mathcal{U}} J_a(\QQ,\xi),\label{eq.agentProblem.def.3}\\
&J_a(\QQ,\xi):=\EE^\QQ U_a\left(\xi(X),X,M\right),\nonumber
\end{align}
where $U_a$ is introduced in the previous section.

Finally, we introduce
\begin{align*}
&\mathcal{U}^*(\xi) :=\{\nu \in \mathcal{U} : J_a\left(\nu,\xi\right) = V_a(\xi)\in\RR\},\\
&\mathcal{C}^{a} := \{\xi\in\mathcal{C}:\, V_a(\xi)\geq R_a,\, \mathcal{U}^*(\xi) \neq \emptyset\},
\end{align*}
and define the relaxed problem of the principal:
\begin{align}
&V_p:=\sup_{\xi\in\mathcal{C}^{a}} J_p(\xi)\label{eq.principalProblem.def.relaxed}\\
&J_p(\xi):=\sup_{\QQ\in\mathcal{U}^*(\xi)}\EE^\QQ U_p\left(\xi(X),X,M\right),\nonumber
\end{align}
where $ \bar{U}_p$ is defined in the previous section.

\smallskip

The assumptions that follow ensure that the associated expectations are well defined for all $\QQ\in\mathcal{U}$ and all $\xi\in\mathcal{C}$.


\subsection{Solution to the relaxed problem}
In order to show the existence of an optimal contract in the relaxed formulation we need to make some technical assumptions. The following assumption is quite standard.
\begin{ass}\label{ass:1}
We assume that the utility functions satisfy the following:
\begin{itemize}
\item[a)] Continuity: $U_a$ and $U_p$ are continuous on $\RR\times C([0,T],\RR^k)\times(0,\infty)$.
\item[b)] Feasibility: for any $\xi\in\mathcal{C}$ and $\QQ\in\mathcal{U}$ the expectations $\EE^\QQ U_{a/p}\left(\xi(X),X,M\right)$ are well defined (even if infinite), and there exist $\xi\in\mathcal{C}$ and $\QQ\in\mathcal{U}^*(\xi)$ such that $\EE^\QQ\left[U_a(\xi(X),X,M)\right]\geq R_a$.
\end{itemize}
\end{ass}

\smallskip

The next assumption ensures the appropriate uniform integrability. To explain the need for this assumption, we note that, although the set $\mathcal{U}$ is pre-compact in the weak topology, it is not closed. Indeed, the conditions 1, 2 and 4 of Definition \ref{def:U} are not preserved by the weak limits (remarkably, the last condition is stable under the weak limits, as shown in the proof of Theorem \ref{existence:weak:opt}). The following assumption guarantees the existence of a subset $\mathcal{U}_0\subset\mathcal{U}$ which contains all optimal responses of the agent and whose closure w.r.t. the weak topology is contained in $\mathcal{U}$.

\begin{ass}\label{ass:1.5}
We assume that there exists a set $\mathcal{U}_0\subset\mathcal{U}$, such that, for any $\xi\in\mathcal{C}$,
\begin{align*}
& \sup_{\QQ\in\mathcal{U}} \EE^{\QQ}\left[ U_a(\xi(X),X,M) \right]= \sup_{\QQ\in\mathcal{U}_0} \EE^{\QQ}\left[ U_a(\xi(X),X,M)\right]
\end{align*}
and $\mathcal{U}^*(\xi)\subset\mathcal{U}_0$, and such that the following holds:
\begin{itemize}
\item[a)] Uniform integrability: for any $\epsilon>0$, there exists $k\geq0$ such that, for all $\xi\in\mathcal{C}$ and all $\QQ\in\mathcal{U}_0$,
\begin{align*}  
&  \EE^{\QQ} \left[|U_{a}(\xi(X),X,M)|\,\mathbbm{1}_{\{|U_a(\xi(X),X,M)|\geq k\}} \right]\leq \varepsilon.
\end{align*} 
\item[b)] Upper uniform integrability: for any $\varepsilon>0$, there exists $k\geq0$ such that, for all $\xi\in\mathcal{C}$ and all $\QQ \in \mathcal{U}_0$, 
\begin{align*}
& \EE^{\QQ} \left[U_{p}(\xi(X),X,M)\,\mathbbm{1}_{\{U_{p}(\xi(X),X,M)\geq k\}} \right]\leq \varepsilon.
\end{align*}
\item[c)] Uniform integrability of $M$: for any $\varepsilon>0$, there exists $k\geq0$ such that, for all $\QQ\in\mathcal{U}_0$,
\begin{align*}
& \QQ(M\leq 1/k) \leq \varepsilon.
\end{align*}
\item[d)] There exists $C_0> 0$ such that, for all $\QQ\in\mathcal{U}_0$,
\begin{align*}
& \EE^{\QQ} \left[M\log(M)\right] \leq C_0. 
\end{align*}
\end{itemize}
\end{ass}

\smallskip

Let us comment on how one can verify the above assumption. Following the discussion in Section \ref{strong:opt:contract}, assume, for example, that $U_{a/p}(\xi,x,m)=\tilde U_{a/p}(\xi,x)\,m^q$, with $q\in(0,1)$ and with bounded positive $\tilde U_{a/p}$. Then, the Cauchy inequality, Vall\'ee-Poussin theorem, and the condition $\EE^\QQ M=1$, yield the properties (a), (b) of Assumption \ref{ass:1.5}, even with $\mathcal{U}_0=\mathcal{U}$. Using the strict convexity of $U_{a/p}$ and Jensen's inequality, we deduce (with the help of Lemma \ref{le:strong.sol.le1}) that every optimal response of the agent is in the form $\hat\QQ^\nu(dx,dm):=\mathbb{W}(dx)\,\delta_{\bar M(\nu)}(dm)$, for some $\nu\in\mathcal{A}$. Then, using the properties of the latter set (note that one can enforce stochastic or deterministic bounds on $\nu$ via a choice of $A,b$), one can verify the properties (c), (d) of Assumption \ref{ass:1.5}. The latter is illustrated in Section \ref{opt:brok:as:info}, where it is also shown how to handle the case $U_{a/p}(\xi,x,m)=\tilde U_{a/p}(\xi,x)\,m - 2\phi\,m\log m$, even with unbounded $\tilde U_{a/p}$.

\medskip

Before stating Theorem \ref{existence:weak:opt} (the main result of this section) it is convenient to introduce the following lemma. 
Although this lemma is, essentially, known, we are not aware of any reference with the precise statement needed herein and, therefore, provide the proof of the lemma in the  Appendix.
Herein, we denote by $\mathbb{H}_{loc}^2([0,T])$ the set of locally square integrable progressively measurable processes on the stochastic basis $(\Omega,\mathbb{F},\WW)$. In addition, for any random element on $(\Omega,\mathcal{F}_T)$ we will view it (whenever needed) as a random element on $(\Omega\times\RR,\mathcal{F}_T\otimes\mathcal{B}(\RR))$ without stating it explicitly. 

\begin{lemma}\label{lemma:condexpM}
Let $\QQ$ be a probability measure on $(\Omega\times\RR,\mathcal{F}_T\otimes\mathcal{B}(\RR))$ that satisfies conditions 1--4 in Definition \ref{def:U}. Then, there exists $\nu \in \mathbb{H}_{loc}^2([0,T])$ such that 
\begin{equation*}
    \EE^{\QQ}\left( M | \mathcal{F}_T \right) = \bar{M}_T(\nu),\quad \QQ-a.s.
\end{equation*}
Moreover, $\nu$ satisfies: 
\begin{equation*}
    \EE^{\QQ}\left[M\int_0^T\|\nu_s\|^2ds \right] =  2\EE^{\WW}\left[\bar{M}_T(\nu)\log\left( \bar{M}_T(\nu)\right) \right] < \infty. 
\end{equation*}
\end{lemma}

In view of the feasibility assumption, there exists a sequence $\xi^n\in\mathcal{C}^a$ s.t. $$\lim_{n\rightarrow\infty}J_p(\xi^n)=V_p.$$
Consider the associated $\QQ^n\in \mathcal{U}^*(\xi^n)$. In the next theorem we show that the sequence $\{\xi_n\}$ contains a limit point $\xi^*$ which is an optimal contract in the relaxed formulation. 

\begin{theorem}\label{existence:weak:opt}
Under Assumptions \ref{ass:1} and \ref{ass:1.5}, there exist $\xi^*\in \mathcal{C}^a$ and a probability measure $\QQ^*\in\mathcal{U}$, such that $\xi^n$ and $\QQ^n$ converge along a subsequence, respectively, to $\xi^*$ (uniformly on compacts) and to $\QQ^*$ (in the weak topology). Moreover, $\xi^*$ is an optimal contract and $\QQ^*\in \mathcal{U}^*(\xi^*) $.
\end{theorem}
\begin{proof}
We will show that the sequence $(\xi^n)_{n\in \mathbb{N}} \subset \mathcal{C}$ is relatively compact. By definition, $(\xi^n)_{n\in \mathbb{N}}$ is an equicontinuous family. In addition, for every $x \in C([0,T],\RR^k)$ the sequence $(\xi^n(x))_{n\in \mathbb{N}}$ is bounded by definition. Then, by Ascoli's theorem applied on the topology of the compact convergence (\cite{Kelley}, Theorem 18) there exists a subsequence uniformly convergent on compact sets. Without lost of generality we assume that $(\xi^n)_{n=1}^\infty$ is convergent and $\xi^*$ is its limit. The later implies that for any compact set $K\subset C([0,T],\RR^k)$:
\begin{equation*}
    \lim_{n \rightarrow \infty} \sup_{x \in K} |\xi^n(x)-\xi^*(x) | = 0.
\end{equation*}
Using that $\mathcal{C}$ is closed with respect to the topology of the compact convergence we obtain that  $\xi^* \in \mathcal{C}$.
\\
On the other hand, we will show that the sequence $(\QQ^n)_{n\in \NN}$ has a convergent subsequence. It is enough to show that the sequence is tight. We observe that the $(1)$-marginal sequence is tight as it is the Wiener measure for every $n\in \mathbb{N}$. The $(2)$-marginal sequence is tight as well. Indeed, by Markov's inequality and property 1 in Definition \ref{def:U}:
\begin{equation*}
    \lim_{N \rightarrow \infty}\sup_{n \in \mathbb{N}} \mathbb{Q}^n\left(|M| \geq N \right) \leq \lim_{N\rightarrow \infty}\frac{1}{N} = 0. 
\end{equation*}
Hence, using Prohorov's theorem (\cite{Billingsley}, Theorem 5.1) there exits a weakly convergent subsequence $\left(\QQ^{n_k} \right)_{k \in \mathbb{N}}$. We denote $\QQ^*$ its limit. Without lost of generality we will denote $(\QQ^n)_{n\in \NN}$ the convergent subsequence. Using Assumption \ref{ass:1.5} we can choose a sequence such that $(\QQ^n)_{n\in\NN}\subset \mathcal{U}_0$. 

It remains to show that $\QQ^* \in \mathcal{U}$. First, we will show that $\QQ^*$ satisfies property $1$ of definition $2$.
Indeed, using Assumption \ref{ass:1.5} we have that for all $\epsilon >0$ there exits $N_1,N_2 \in \NN$ such that
\begin{equation*}
    \begin{split}
    \limsup_{n \rightarrow \infty} \EE^{\QQ^n}M  &\leq  \limsup_{n \rightarrow \infty} \EE^{\QQ^n}\left[M\wedge N_1\right] + \limsup_{n \rightarrow \infty} {\EE}^{\QQ^n}\left[\left(M-N_1\right)^{+}\right] \\
    &\leq  {\EE}^{\QQ^*} \left[M \wedge N_1\right]+\epsilon \leq {\EE}^{\QQ^*} M  + \epsilon. 
    \end{split},
    \end{equation*}
    And
 \begin{equation*}
    \begin{split}
    \liminf_{n \rightarrow \infty} \EE^{\QQ^n}M  &\geq \liminf_{n \rightarrow \infty} {\EE}^{\QQ^n}\left[M\vee (-N_2)\right]+ \liminf_{n \rightarrow \infty} {\EE}^{\QQ^n}\left[\left(M+N_2\right)^{-} \right] \\
    &\geq {\EE^{\QQ^*}} \left[M \vee (-N_2)\right]- \epsilon \geq \EE^{\QQ^*} M  - \epsilon. 
    \end{split}
    \end{equation*}   
    Hence, combining the previous two equations we obtain: 
    \begin{equation*}
        \lim_{n \rightarrow \infty } \EE^{\QQ^n} M = \EE^{\QQ^*} M. 
    \end{equation*}
Now will well prove that $\QQ^*$ satisfies condition $2$ in definition \ref{def:U}. Using Portmanteau theorem we have that : 
\begin{equation*}
   \QQ^*(M \geq 0 ) \geq  \limsup_{n \rightarrow \infty} \QQ^n(M \geq 0) = 1.
\end{equation*}
By a similar argument, using Assumption \ref{ass:1.5}.$c$ and Portmanteau theorem we obtain that for every $\epsilon > 0$, there exits $k > 0$ such that:
\begin{equation*}
   \QQ^*(M = 0) \leq \QQ^*(M < 1/k ) \leq \limsup_{n \rightarrow \infty}\QQ^n(M < 1/k) \leq \sup_{n \in \mathbb{N}}\QQ^n(M < 1/k) < \epsilon. 
\end{equation*}
Hence, $\QQ^*(M > 0 ) = 1$. It is trivial to check that $\QQ^*$ satisfies property $3$ in Definition \ref{def:U}. Using that $(\QQ^n)_{n\in \mathbb{N}} \subset \mathcal{U}_0$ and Portmanteau theorem we obtain:
\begin{equation*}  
    \EE^{\QQ^*}[ M\log(M)] \leq \liminf_{n \rightarrow \infty}\EE^{\QQ^n}[ M\log(M)] < \infty.  
\end{equation*}
Therefore, $\QQ^*$ satisfies property $4$ in Definition \ref{def:U}. Finally, we just need to show that $\QQ^*$ satisfies property $5$ in definition \ref{def:U}.
We will show that the 
\begin{equation*}
    \lim_{n\rightarrow \infty}\EE^{\QQ^n}\left[M\eta^{\top} \int_{s\wedge\tau_N}^{t\wedge\tau_N}\left(b_rdr+A_rdX_r\right) \right] = \EE^{\QQ^*}\left[M\eta^{\top} \int_{s\wedge\tau_N}^{t\wedge\tau_N}\left(b_rdr+A_rdX_r\right) \right]. 
\end{equation*}
In order to prove the previous expression we will show it separately for the finte variation part and the stochastic integral. We
introduce the probability measures $\tilde{\QQ}^n$ defined by:
\begin{equation*}
    d\tilde{\QQ}^n := Md\QQ^n. 
\end{equation*}
Firstly, we will show that $\tilde{\QQ}^n$ converges weakly to the measure $\tilde{\QQ}^{\infty}$ defined by
\begin{equation*}
      d\tilde{\QQ}^{\infty} := Md\QQ^*. 
\end{equation*}
Indeed, let $f \in C_b\left(C\left([0,T],\RR^k\right)\times \RR\right)$. We observe: 
\begin{equation*}
\lim_{N\rightarrow\infty}\sup_{n\in\NN}\EE^{\QQ^n}\left[|Mf(X,W)|\mathbbm{1}_{|Mf(X,W)|>N}\right] \leq K\lim_{N\rightarrow\infty}\sup_{n\in\NN}\EE^{\QQ^n}\left[M\mathbbm{1}_{M>|\frac{N}{K}|}\right], \end{equation*}
where $K = 1+\sup_{(x,m)\in C([0,T],\RR^k),\RR }|f(x,m)|$.

Using Assumption \ref{ass:1.5}.$c$ we have that the previous limit converges to zero.
Hence, for all $\epsilon>0$ there exit $N_1, N_2 \in \NN$ such that: 
\begin{equation*}    
\begin{split}
&\limsup_{n\rightarrow \infty}\EE^{\tilde{\QQ}^n}\left[f(X,M)\right] = \limsup_{n\rightarrow \infty}\EE^{\QQ^n}\left[Mf(X,M)\right] \\
&\leq \limsup_{n\rightarrow\infty}\EE^{\QQ^n}\left[(Mf(X,M)-N_1)^{+}\right]+\limsup_{n\rightarrow\infty}\EE^{\QQ^n}\left[(Mf(X,M))\wedge N_1\right] \\
&\leq \epsilon + \EE^{\QQ^*}\left[(Mf(X,M))\wedge N_1\right] \leq  \EE^{\QQ^*}\left[Mf(X,M)\right]+\epsilon
=\EE^{\tilde{\QQ}^{\infty}}\left[f(X,M)\right]+\epsilon
\end{split}
\end{equation*}
and
\begin{equation*}    
\begin{split}
&\liminf_{n\rightarrow \infty}\EE^{\tilde{\QQ}^n}\left[f(X,M)\right] = \liminf_{n\rightarrow \infty}\EE^{\QQ^n}\left[Mf(X,M)\right] \\
&\geq \liminf_{n\rightarrow\infty}\EE^{\QQ^n}\left[(Mf(X,M)+N_2)^{-}\right]+\liminf_{n\rightarrow\infty}\EE^{\QQ^n}\left[(Mf(X,M))\vee (-N_2)\right] \\
&\geq -\epsilon + \EE^{\QQ^*}\left[(Mf(X,M))\vee (-N_2)\right] \geq  \EE^{\QQ^*}\left[Mf(X,M)\right]-\epsilon
= \EE^{\tilde{\QQ}^{\infty}}\left[f(X,M)\right]-\epsilon
\end{split}
\end{equation*}
The later implies
\begin{equation*}
    \lim_{n\rightarrow \infty}\EE^{\tilde{\QQ}^n}\left[f(X,M) \right] = \EE^{\tilde{\QQ}^{\infty}}\left[f(X,M) \right].
\end{equation*}
Therefore, 
$(\tilde{\QQ}^n)_{n=1}^\infty$ converges weakly to $\tilde{\QQ}^{\infty}$.
\\
Moreover, we observe that for all $f \in C_b\left(C([0,T],\RR^k),\RR\right)$, and $n \in \NN \sqcup \{ \infty\}$:
\begin{equation*}
    \EE^{\tilde{\QQ}^n}\left[f(X) \right] = \EE^{\QQ^n}\left[M f(X) \right] = \EE^{\WW}\left[\EE^{\QQ^n}(M|\mathcal{F}_T^X)f(X) \right].
\end{equation*}
Using Lemma \ref{lemma:condexpM} we have that there exits $\nu^n \in \mathbb{H}^2_{loc}([0,T])$ such that
\begin{equation*}
    \EE^{\QQ^n}(M|\mathcal{F}_T^X) = \bar{M}_T(\nu^n), \quad n \in \NN \sqcup \{\infty\}.
\end{equation*}
Therefore, 
\begin{equation}\label{equiv.tildeQ}
d\tilde{\QQ}^n|_{\mathcal{F}_T^X} = M_T(\nu^n)d\WW,\quad n \in \NN \sqcup \{ \infty \}.
\end{equation}
\\
We introduce the mapping $T: C([0,T],\RR^k) \mapsto \RR$ defined by
\begin{equation*}
    T(x) = \int_{s\wedge\tau_N(x)}^{t\wedge\tau_N(x)}\eta(x)^{\top}b_{r}(x)dr.
\end{equation*}
We observe that $T$ is bounded as $\eta$ and $b_{.\wedge\tau_N}$ are bounded. Moreover, using Lemma $5.9$ in \cite{DeLaRue}: 
\begin{equation*}
    \WW\left( \{x \in C([0,T],\RR^k) : \tau_N \text{ is continuous in } x \} \right) = 1. 
\end{equation*}
By \eqref{equiv.tildeQ} we have: 
\begin{equation}\label{crossing.Qtilde}
     \tilde{\QQ}^{\infty}\left( \{x \in C([0,T],\RR^k) : \tau_N \text{ is continuous in } x \} \right) = 1. 
\end{equation}
The last observation implies that the set of discontinuities of $T$ has $\tilde{\QQ}^\infty $-measure $0$. Hence, by the continuous mapping theorem we obtain: 
\begin{equation}\label{eq:limitb}
    \begin{split}
    &\lim_{n\rightarrow \infty}  \EE^{\QQ^n}\left[M\eta^{\top} \int_{s \wedge\tau_N}^{t\wedge\tau_N}b_rdr\right] = \lim_{n\rightarrow\infty}\EE^{\tilde{\QQ}^n}\left[\eta^{\top} \int_{s\wedge\tau_N}^{t\wedge\tau_N}b_rdr\right] \\  
    &=\EE^{\tilde{\QQ}}\left[\eta^{\top} \int_{s\wedge\tau_N}^{t\wedge\tau_N}b_rdr\right] =\EE^{\tilde{\QQ}^*}\left[M\eta^{\top} \int_{s\wedge\tau_N}^{t\wedge\tau_N}b_rdr\right].
    \end{split}
\end{equation}
Now, we will show that the stochastic integrals converge. Applying Skorokhod's representation theorem (\cite{Billingsley}, Theorem 6.7) there exists a sequence $(\hat{X}^n)_{n\in \NN}$ converging a.s. to $\hat{X}^\infty$ in some probability space $(\hat\Omega,\hat{\mathcal{F}},\hat{\PP})$, where $\hat\PP\circ (\hat{X}^n)^{-1}=\tilde{\QQ}^n$ and $\hat\PP\circ (\hat{X}^\infty)^{-1}=\tilde{\QQ}^\infty$. As the mapping $A$ is continuous we have that $\left(\hat{X}^n,A(\hat{X}^n)\right)$ converges $\hat{\PP}$ a.s. to $\left(\hat{X}^\infty,A(\hat{X}^\infty)\right)$. Using Lemma \ref{lemma:condexpM}, we obtain that for every $n \in \NN$ there exist $\FF$-adapted processes $\nu^n \in \mathbb{H}^2_{loc}([0,T])$ that satisfy
\begin{equation*}
    \begin{split}
    \EE^{\QQ^n}\left(M|\mathcal{F}_T^X \right) = \bar{M}_T(\nu^n),
   \quad \EE^{\QQ^*}\left(M|\mathcal{F}_T^X \right) = \bar{M}_T(\nu^\infty).
    \end{split}
\end{equation*}
Using \eqref{equiv.tildeQ}, the later implies by Girsanov theorem (\cite{KaratzasShreve}, Theorem 5.1):
\begin{equation*}
    \begin{split}
    dX_t &= \nu^n_t dt + dB^n_t, \quad\tilde{\QQ}^n- a.s., \\
    dX_t &=  \nu^\infty_t dt + dB^\infty_t, \quad \tilde{\QQ}^\infty-a.s.,
     \end{split}
\end{equation*}
where $B^n$ is a $(\Omega,\tilde{\QQ}^n,\FF)$ Brownian motion.
Using Skorokhod's representation theorem (\cite{Billingsley}, Theorem $6.7$) we have
\begin{equation*}
    \begin{split}
         d\hat{X}^n_t &= \hat{\nu}^n_t dt + d\hat{B}^n_t, \quad \hat{\PP}-a.s., \\
         d\hat{X}^{\infty}_t &= \hat{\nu}^\infty_t dt + d\hat{B}^\infty_t, \quad \hat{\PP}-a.s.,
    \end{split}
\end{equation*}
where, for all $t \in [0,T]$, $\hat{\nu}^n_t:=\nu^n_t \circ \hat{X}^n$, and $\hat{B}^n_t:=B^n_t \circ \hat{X}^n$ are standard brownian motions defined on $(\hat{\Omega},\hat{\PP} ,\mathbb{F}^{\hat{X}^n})$.
Moreover, we observe that 
\begin{equation*}
\begin{split}
 &\hat{\EE}\left[ \int_0^T |\hat{\nu}^n_s|ds\right] \leq T^{1/2} \left(\hat{\EE}\left[\int_0^T (\hat{\nu}^n_s)^2ds\right]\right)^{1/2} 
 = (2T)^{1/2}\left(\EE^{\tilde{\QQ}^n}\log\bar{M}_T(\nu^n) \right)^{1/2} \\
&= (2T)^{1/2}\left(\EE^{\QQ^n}\left[\bar{M}_T(\nu^n)\log\bar{M}_T(\nu^n)\right]\right)^{1/2} 
\leq (2T)^{1/2}\left(\EE^{\WW} \left[M \log(M)\right]\right)^{1/2} \leq \sqrt{2T C_0},
\end{split}
\end{equation*}
where we used Jensen's inequality, Cauchy-Schwartz inequality and Assumption \ref{ass:1.5}.c.
We also observe: 
\begin{equation*}
    \hat{\EE} \left([\hat{B}^n]_T \right) = T, \quad n \in \NN \sqcup \{\infty\}.
\end{equation*}
Hence, applying (\cite{ProtterKurtz}, Theorem 2.2), we have that the processes $(\hat{X}^n,\int_0^. A_s(\hat{X}^n_s)d\hat{X}^n_s)$ converge in probability $\hat{\PP}$ to $\left(\hat{X}^\infty, \int_0^.A_s(\hat{X}^\infty_s)d\hat{X}^\infty_s\right)$. 
Then, defining the processes $\hat{L}^n:=\int_0^. A_r(\hat{X}^n)d\hat{X}^n_r$, $n \in \NN \sqcup \{ \infty\}$, we have for all $\epsilon > 0$: 
\begin{equation*}
    \lim_{n\rightarrow\infty} \hat{\PP} \left(\|\hat{L}^n-\hat{L}^\infty \|_{C([0,T],\RR^k)} \geq \epsilon \right) = 0.  
\end{equation*}
Using again Lemma $5.9$ in \cite{DeLaRue}:
\begin{equation*}
     \tilde{\QQ}^{\infty}\left( \{x \in C([0,T],\RR^k) : \tau_N \text{ is continuous in } x \} \right) = 1,
\end{equation*}
which implies
\begin{equation*}
     \hat{\PP}\left( \{\omega \in \hat{\Omega} : \tau_N \text{ is continuous in }\hat{X}^\infty(\omega)\} \right) = 1.
\end{equation*}
Therefore, by the continuous mapping theorem:
\begin{equation*}
    \lim_{n\rightarrow \infty} \hat{\PP} \left(\vert\hat{L}^n_{t\wedge\tau_N(\hat{X}^n)}-\hat{L}^\infty_{t\wedge\tau_N(\hat{X}^\infty)}\vert \geq \epsilon \right) = 0. 
\end{equation*}
The later implies that $\hat{L}^n_{t\wedge \tau_N(\hat{X}^n)}$ converges in probability $\hat{\PP}$ to $\hat{L}^{\infty}_{t\wedge\tau_N(\hat{X}^\infty)}$. Analogously, $\hat{L}^{n}_{s\wedge\tau_N(\hat{X}^n)}$ converges to $\hat{L}^{\infty}_{s\wedge\tau_N(\hat{X}^\infty)}$ in probability $\hat{\PP}$. Therefore, 
$\hat{L}^{n}_{t\wedge\tau_N(\hat{X}^n)}-\hat{L}^{n}_{s\wedge\tau_N(\hat{X}^n)} $ converges in probability $\hat{\PP}$ to $\hat{L}^{\infty}_{t\wedge\tau_N(\hat{X}^\infty)}-\hat{L}^{\infty}_{s\wedge\tau_N(\hat{X}^\infty)} $ . Taking a subsequence that converge $\hat{\PP}$-a.s. we obtain:  
\begin{equation}\label{eqn:prop4}
    \begin{split}   
    &\EE^{\QQ^*}\left[M\eta(X)^{\top}\int_{s\wedge\tau_N}^{t\wedge \tau_N}A_rdX_r \right] =  \EE^{\tilde{\QQ}^\infty}\left[\eta(X)^{\top}\int_{s\wedge \tau_N}^{t\wedge\tau_N}A_rdX_r \right] \\
    &=\hat{\EE}\left[\eta(\hat{X}^\infty)^{\top} \int_{s\wedge \tau_N(\hat{X}^\infty)}^{t\wedge \tau_N(\hat{X}^\infty)}A_r(\hat{X}^\infty)d\hat{X}^{\infty}_r \right] 
    = \hat{\EE}\left[\lim_{n\rightarrow \infty}\eta(\hat{X}^n)^{\top} \int_{s\wedge\tau_N(\hat{X}^n)}^{t\wedge\tau_N(\hat{X}^n)}A_r(\hat{X}^n)d\hat{X}^n_r \right].
    \end{split}
\end{equation}

Finally, we observe that the sequence $\left(\eta(\hat{X}^n)^{\top}\int_{s\wedge\tau_N(\hat{X}^n)}^{t\wedge\tau_N(\hat{X}^n)}A_r(\hat{X}^n)d\hat{X}^n_r\right)_{n\in \mathbb{N}}$ is $\hat{\PP}$ uniformly integrable. Indeed, it is enough to check that the sequence is bounded in
$L^2(\hat{\Omega},\hat{\mathcal{F}},\hat{\PP})$. By Ito's Isometry: 
\begin{equation*}
\begin{split}
&\sup_{n\in \mathbb{N}}\hat{\EE}\left[\left(\int_{s\wedge\tau_N(\hat{X}^n)}^{t\wedge\tau_N(\hat{X}^n)}\eta(\hat{X}^n)^{\top}A_r(\hat{X}^n)d\hat{X}^n_r \right)^2 \right] 
= \sup_{n\in \mathbb{N}}\EE^{\tilde{\QQ}^n}\left[\left(\int_{s\wedge\tau_N}^{t\wedge\tau_N}\eta(X)^{\top}A_r(X)dX_r \right)^2 \right] \\
&= \sup_{n\in \mathbb{N}}\EE^{\tilde{\QQ}^n}\left[\int_{s\wedge\tau_N}^{t\wedge \tau_N}|\eta^\top A_r\nu^n_r|^2dr \right] 
\leq \|\eta^\top A_{.\tau_N}\|^2_{C([0,T],\RR^k)}\sup_{n\in \mathbb{N}}\EE^{\tilde{\QQ}^n}\left[\int_{0}^{T}\| \nu_r^n\|^2dr \right] \\
&=  2\|\eta^\top A_{.\tau_N}\|_{C([0,T],\RR^k)}\sup_{n\in \mathbb{N}}\EE^{\WW}\left[\bar{M}_{T}(\nu^n)\log(\bar{M}_{T}(\nu^n)\right]
\leq 2\|\eta^\top A_{.\tau_N}\|_{C([0,T],\RR^k)}\sup_{n\in \mathbb{N}}\EE^{\QQ^n}\left[M\log(M) \right] < \infty, 
\end{split}
\end{equation*}
where we used Cauchy-Schwartz inequality, Lemma \ref{lemma:condexpM}, Jensen's inequality, and the inclusion $(\QQ^n)_{n\in \NN} \subset \mathcal{U}_0$. 
Therefore, using \eqref{eqn:prop4}:
\begin{equation*}
    \begin{split}
        &\EE^{\QQ^*}\left[M\eta^{\top} \int_s^{t}A_rdX_r \right] = \hat{\EE}\left[\lim_{n\rightarrow \infty}\eta(\hat{X}^n)^{\top}\int_{s\wedge\tau_N(\hat{X}^n)}^{t\wedge\tau_N(\hat{X}^n)}A_r(\hat{X}^n)d\hat{X}^n_r \right]\\
        &= \lim_{n\rightarrow \infty}\hat{\EE}  \left[\eta(\hat{X}^n)^{\top} \int_{s\wedge \tau_N(\hat{X}^n)}^{t\wedge \tau_N(\hat{X}^n)}A_r(\hat{X}^n)d\hat{X}^n_r \right] 
        = \lim_{n\rightarrow \infty}\EE^{\QQ^n}  \left[M\eta^{\top} \int_{s\wedge\tau_N}^{t\wedge\tau_N}A_r(X)dX_r \right].
    \end{split}
\end{equation*}
Combining the last equation with \eqref{eq:limitb} we obtain: 
\begin{equation*}
   \EE^{\QQ^*}\left[M\eta^{\top} (Z_{t\wedge\tau_N}-Z_{s\wedge\tau_N}) \right] = \lim_{n\rightarrow \infty}  \EE^{\QQ^n}\left[M\eta^{\top} (Z_{t\wedge\tau_N}-Z_{s\wedge\tau_N}) \right] \leq 0. 
\end{equation*}
Hence, $\QQ^* \in \mathcal{U}$.

Without lost of generality, we assume that the convergence occurs along the original sequence.

\medskip

Let us show that $\xi^*\in \mathcal{C}^a$, under Assumptions \ref{ass:1} and \ref{ass:1.5}. To this end, we apply the Skorokhod's representation theorem to obtain a sequence $\{(\check X^n,\check M^n)\}$ converging a.s., on some probability space $(\check\Omega,\check{\mathcal{F}},\check{\PP})$, to $(\check{X}^\infty,\check{M}^\infty)$, such that $\check\PP\circ (\check{X}^n,\check{M}^n)^{-1}=\QQ^n$ and $\check\PP\circ (\check{X}^\infty,\check{M}^\infty)^{-1}=\QQ^*$. Then, due to Assumption \ref{ass:1.5}.a, for any $\varepsilon>0$ there exists $N\geq0$ such that
\begin{align*}
& R_a\leq \limsup_{n\rightarrow \infty}J_a(\QQ^n,\xi^n)=\limsup_{n\rightarrow \infty}\EE^{\QQ^n} U_a\left(\xi^n(X),X,M\right)
= \limsup_{n\rightarrow \infty}\EE^{\QQ^n} U_a\left(\xi^n(X),X,M\right)\wedge N \\
&+ \limsup_{n\rightarrow \infty}\EE^{\QQ^n} \left[\left(U_a\left(\xi^n(X),X,M\right)-N\right)^+\right]
\leq \limsup_{n\rightarrow \infty}\check{\EE}\left[ U_a\left(\xi^n(\check{X}^n),\check{X}^n,\check{M}^n\right)\wedge N\right] + \varepsilon.
\end{align*}
Using Fatou's lemma, we obtain:
\begin{align*}
&\limsup_{n\rightarrow\infty}\check{\EE} \left[U_a\left(\xi^n(\check X^n),\check X^n,\check M^n\right)\wedge N\right]
\leq \check{\EE} \left[ U_a\left(\xi^*(\check X^\infty),\check X^\infty,\check M^\infty\right)\wedge N\right]\\
&= \check{\EE}\left[U_a\left(\xi^*(\check X^\infty),\check X^\infty,\check M^\infty\right)\right]
= \EE^{\QQ^*} U_a\left(\xi^*(X),X,M\right) = J_a(\QQ^*,\xi^*).
\end{align*}
Thus, we obtain $J(\QQ^*,\xi^*)\geq R_a-\varepsilon$, for any $\varepsilon>0$, which yields $V_a(\xi^*)\geq R_a$ and in turn $\xi^*\in \mathcal{C}^a$.

\smallskip

Next, we show that $\QQ^*\in\mathcal{U}^*(\xi^*)$.
To this end, we assume the contrary: i.e., that there exists $\tilde\QQ\in\mathcal{U}$ such that $J_a(\QQ^*,\xi^*)< J_a(\tilde\QQ,\xi^*)$.
The derivation in the preceding paragraph and the optimality of $\QQ^n$ yield
\begin{align*}
& J_a(\QQ^*,\xi^*) \geq \limsup_{n\rightarrow\infty} J_a(\QQ^n,\xi^n)
\geq \limsup_{n\rightarrow\infty} J_a(\tilde\QQ,\xi^n).
\end{align*}
In addition, due to Assumption \ref{ass:1.5}.$a$, for any $\varepsilon>0$ there exists $N\geq0$ such that
\begin{align*}
& \liminf_{n\rightarrow \infty}J_a(\tilde\QQ,\xi^n)=\liminf_{n\rightarrow \infty}\EE^{\tilde\QQ} U_a\left(\xi^n(X),X,M\right)
\geq \liminf_{n\rightarrow \infty}\EE^{\tilde\QQ}  U_a\left(\xi^n(X),X,M\right)\vee (-N) \\
&+ \liminf_{n\rightarrow \infty}\EE^{\tilde\QQ}\left[ \left( U_a\left(\xi^n(X),X,M\right)+N\right)^-\right]
\geq \liminf_{n\rightarrow \infty}\EE^{\tilde\QQ}U_a\left(\xi^n(X),X,M\right)\vee (-N) - \varepsilon.
\end{align*}
Using Fatou's lemma, we obtain:
\begin{align*}
&\liminf_{n\rightarrow\infty} \EE^{\tilde\QQ}   U_a\left(\xi^n(X),X,M\right)\vee (-N)
\geq \EE^{\tilde\QQ} \bar U_a\left(\xi^*(X),X,M\right)\vee (-N)\\
&\geq \EE^{\tilde\QQ} \bar U_a\left(\xi^*(X),X,M\right) = J_a(\tilde\QQ,\xi^*).
\end{align*}
The above yields $\limsup_{n\rightarrow\infty} J_a(\tilde\QQ,\xi^n)\geq J_a(\tilde\QQ,\xi^*)$ and leads to the desired contradiction.
\\
\smallskip
Finally, we show that $J_p(\xi^*)=V_p$. Indeed, due to Assumption \ref{ass:1.5}.b, for any $\varepsilon>0$, there exists $N\geq0$ such that
\begin{align*}
&V_p = \lim_{n\rightarrow \infty} J_p(\xi^n) \leq \limsup_{n\rightarrow\infty}\EE^{\QQ^n} U_p\left(\xi^n(X),X,M\right)= \limsup_{n\rightarrow\infty}\EE^{\QQ^n} 
U_p\left(\xi^n(X),X,M\right)\wedge N \\
&+ \limsup_{n\rightarrow\infty}\EE^{\QQ^n} \left(U_p\left(\xi^n(X),X,M\right)-N\right)^+
\leq \limsup_{n\rightarrow\infty}\check{\EE} U_p\left(\xi^n(\check X^n),\check X^n,\check M^n\right)\wedge N + \varepsilon.
\end{align*}

Using Fatou's lemma we obtain:
\begin{align*}
&\limsup_{n\rightarrow\infty}\check{\EE} U_p\left(\xi^n(\check X^n),\check X^n,\check M^n\right)\wedge N
\leq \check{\EE} U_p\left(\xi^*(\check X^ \infty),\check X^\infty,\check M^\infty\right)\wedge N\\
&= \EE^{\QQ^*} U_p\left(\xi^*(X),X,M\right)\wedge N
\leq
\EE^{\QQ^*} U_p\left(\xi^*(X),X,M\right) \leq J_p(\xi^*),
\end{align*}
which yields the desired $J_p(\xi^*)=V_p$.

\smallskip

The above analysis proves that $\xi^*$ is an optimal contract and that $\QQ^*$ is an optimal response of the agent.
\qed
\end{proof}

\section{Solution to the problem in strong formulation}
\label{solution_strong_problem}

The goal of this section is to show, under an additional concavity assumption, any optimal contact $\xi^*$ in the relaxed problem is also optimal in the strong formulation. We start with the following auxiliary lemma, which is a natural continuation of Lemma \ref{lemma:condexpM} and which shows that condition 4 in Definition \ref{def:U} yields condition (b) in Definition \ref{def:A}. We recall that $\mathbb{H}_{loc}^2([0,T])$ is the set of locally square integrable progressively measurable processes on the stochastic basis $(\Omega,\mathbb{F},\WW)$.

\begin{lemma}\label{le:strong.sol.le1}
For any $\QQ\in\mathcal{U}$, there exists an $\RR^{k}$-valued process $\nu \in \mathbb{H}_{loc}^2([0,T])$, such that: $\nu\in\mathcal{A}$ and
\begin{align*}
&\EE^{\QQ} \left(M\,\vert\,\mathcal{F}_T\right) = \bar M_{T}(\nu),\quad \QQ-a.s.
\end{align*}
\end{lemma}
\begin{proof}
From Lemma \ref{lemma:condexpM} we have that there exists $\nu \in \mathbb{H}_{loc}^2([0,T])$ such that $\EE^{\QQ}\left(M|\mathcal{F}_T^X\right)= \bar{M}_T(\nu)$. We need to show that $\nu \in \mathcal{A}$. Using $\QQ\in \mathcal{U}(\xi)$, we have by condition 4 in Definition \ref{def:U} and the tower property of conditional expectations:
\begin{equation}\label{admissible_strong}
    \begin{split}
     0 &\geq \EE^{\QQ}\left[\eta^\top\int_{s\wedge\tau_N }^{t\wedge\tau_N}\left(b_rdr+A_{r}dX_r\right)\,M\right]=\EE^{\WW}\left[  \int_{s\wedge\tau_N}^{t\wedge\tau_N}\left( \eta^\top b_rdr+ \eta^\top A_{r}dX_r\right)\,\bar{M}_{T}(\nu)\right],\\
     \end{split}
\end{equation}
for any $\eta$ bounded, non-negative $\mathcal{F}_{s\wedge\tau_N}$ measurable random variable. 
Let $m\in \mathbb{N}$ and, $0 \leq t_1 < t_2 < \ldots < t_m \leq s$. We introduce $A^i$ to be the i-th row of the $\RR^{l\times k}$-valued process $A$ defined in Definition \ref{def:U}.
 Then, we introduce the family of sets: 
\begin{equation*}
    \begin{split}
    &\Lambda^i_m := \left\{C \in \mathcal{B}(\RR^{k\times m}): \right.  \\    &\left.\quad\EE^{\WW}\left[\left(\mathbbm{1}_{(X_{t_1\wedge\tau_N},\cdots,X_{t_m\wedge\tau_N})\in C }\right) \int_{ s\wedge\tau_N}^{ t\wedge\tau_N}\left(b^i_rdr+A^i_{r}dX_r\right)\bar{M}_T(\nu) \right]\leq 0\right\}.
    \end{split}
\end{equation*}
We observe that $\Lambda^i_m$ is a monotone class of sets for every $i \in \{1,\cdots,l\}$. Indeed, we introduce an increasing sequence of sets ${C_n}\in \Lambda^i_m$. 

Using that $(b_{.\wedge\tau_N})$ and  $(A_{.\wedge\tau_N})$ are bounded and Lemma \ref{lemma:condexpM} we obtain from Dominated Convergence Theorem: 
\begin{equation*}
    \begin{split}
    &\EE^\WW \left[\left(\mathbf{\mathbbm{1}}_{(X_{t_1\wedge\tau_N},\cdots,X_{t_m \wedge\tau_N})\in \bigcup_{n=1}^\infty C_n}\right)\int_{ s\wedge\tau_N}^{ t\wedge\tau_N}\left(b^i_rdr+A^i_{r}dX_r\right)\bar{M}_T(\nu) \right]  \\
    &= \lim_{n \rightarrow \infty}  \EE^\WW \left[\left(\mathbbm{1}_{(X_{t_1\wedge \tau_N },\cdots,X_{t_m \wedge \tau_N})\in C_n}\right)\int_{s\wedge\tau_N}^{ t\wedge\tau_N }\left(b^i_rdr+A^i_{r}dX_r\right)\bar{M}_T(\nu) \right] \leq 0. 
     \end{split}
\end{equation*}
Analogously, if $(C_n)_{n\in \mathbb{N}} \subset \Lambda^i_m$ is a decreasing sequence of sets. We have again from Dominated Convergence theorem: 
\begin{equation*}
      \begin{split}
    &\EE^\WW \left[\mathbbm{1}_{(X_{t_1\wedge \tau_N},\cdots,X_{t_m\wedge\tau_N})\in \bigcap_{n=1}^\infty C_n}\int_{ s\wedge\tau_N }^{ t\wedge\tau_N}\left(b^i_rdr+A^i_rdX_r\right)\bar{M}_T(\nu) \right]  \\
    &= \lim_{n \rightarrow \infty}  \EE^\WW \left[\mathbbm{1}_{(X_{t_1\wedge\tau_N},\cdots,X_{t_m\wedge\tau_N})\in C_n}\int_{ s\wedge\tau_N}^{ t\wedge\tau_N}\left(b^i_rdr+A^i_{r}dX_r\right)\bar{M}_T(\nu) \right] \leq 0. 
    \end{split}
\end{equation*}
Hence, $\Lambda^i_m$ is a monotone class of sets in $\mathcal{B}(\RR^{k\times m})$ for every $i \in \{1,\cdots,l\}$. We introduce the set of hypercubes in $\RR^k$: 
\begin{equation*}
    \mathcal{Q}^k :=\left\{\prod_{i=1}^k (x_i,y_i], \quad -\infty \leq x_i < y_i \leq \infty  \right\},
\end{equation*}
where we define by convention  $(x,\infty] := (x,\infty)$, $x \in \mathbb{R}$.
Additionally, we introduce the following collection of sets: 
\begin{equation*}
    \mathcal{B}_0 := \left\{ \prod_{i=1}^m Q_i,  Q_i \in \mathcal{Q}^k  \right\}.
\end{equation*}
 It is trivial to see that $\mathcal{B}_0$ defines an algebra of sets. We will show that $\mathcal{B}_0 \subset \Lambda^i_m$ for every $i \in \{1,\cdots, l\}$. Firstly, we will prove that every  $Q:=\prod_{j=1}^mQ_j \in \mathcal{B}_0$ belongs to $\Lambda^i_m$. We can define a sequence of continuous, non-negative, compactly supported functions $f_n: \mathbb{R}^{k\times m} \mapsto \RR^{\geq 0}$ that converge $\WW$-a.s. to $\mathbbm{1}_{\prod_{j=1}^mQ_j}$. Then, by Dominated Convergence Theorem: 
\begin{equation*}
    \begin{split}
    &\EE^\WW \left[\left(\mathbbm{1}_{(X_{t_1\wedge\tau_N},\cdots,X_{t_m\wedge\tau_N})\in \prod_{j=1}^mQ^j}\right)\int_{ s\wedge\tau_N}^{ t\wedge\tau_N}\left(b^i_rdr+A^i_{r}dX_r\right)\bar{M}_T(\nu) \right]  \\
    &=\EE^\WW \left[\lim_{n \rightarrow \infty}f_n(X_{t_1\wedge\tau_N},\cdots,X_{t_m\wedge\tau_N})\int_{s\wedge\tau_N}^{ t\wedge\tau_N}\left(b^i_rdr+A^i_{r}dX_r\right)\bar{M}_T(\nu) \right] \\
    &= \lim_{n \rightarrow \infty}\EE^\WW \left[f_n(X_{t_1\wedge\tau_N},\cdots,X_{t_m\wedge\tau_N})\int_{ s\wedge\tau_N}^{ t\wedge\tau_N}\left(b^i_rdr+A^i_{r}dX_r\right)\bar{M}_T(\nu) \right] \leq 0.
    \end{split}
\end{equation*}
where we used Definition \ref{def:U} in the last equation. 
Then, $Q \in \Lambda^i_m$ for every $i \in \{1,\cdots,l\}$, which implies that $\mathcal{B}_0 \subset \Lambda^i_m$ for every $i\in \{1,\cdots,l\}$. Finally, by the monotone class theorem (\cite{ProtterJacod}, Theorem 6.2) $\Lambda_m^{i} = \mathcal{B}(\RR^{k\times m})$. Next we introduce the cylindrical $\sigma$-algebra on $C([0,s])$ defined by $\sigma(\mathcal{C}_s)$, where $\mathcal{C}_s$ denotes the set of cylinder sets in $C([0,s],\RR^k)$. We introduce
\begin{equation*}
\Lambda^i:= \left\{C \in \sigma(\mathcal{C}_s): \quad \EE^{\WW}\left[\mathbbm{1}_{X_{.\wedge\tau_N}\in C } \int_{s\wedge\tau_N}^{ t\wedge\tau_N }\left(b^i_rdr+A^i_{r}dX_r\right)\bar{M}_T(\nu) \right]\leq 0\right\}.
\end{equation*}
By using Dominated Convergence theorem and Lemma \ref{lemma:condexpM} it is simple to check that $\Lambda^i$ is a monotone class. Moreover, using that $\mathcal{B}(\RR^{m\times k}) = \Lambda_m^{i}$ we obtain that $\mathcal{C}_s \subset \Lambda^i$. Hence, by the monotone class theorem (\cite{ProtterJacod}, Theorem 6.2) we obtain that $\Lambda^i = \sigma(\mathcal{C}_s)$. Therefore, $\Lambda^i = \mathcal{F}_s$, for every $i \in \{1,\cdots,l\}$. 

We introduce the measure $\tilde{\QQ}$ defined by  
$$d\tilde{\QQ}:= \bar{M}_T(\nu)d\WW.$$ 

Taking $C_i := \left\{  \EE^{\tilde{\QQ}}\left[\int_{s\wedge\tau_N}^{ t\wedge\tau_N}\left(b^i_r+A^i_{r}\nu_r\right)dr|\mathcal{F}_{s\wedge\tau_N}\right]> 0\right\}$, we obtain from Girsanov theorem (\cite{KaratzasShreve}, Theorem 5.1): 
\begin{equation*}
\begin{split}
&\EE^{\tilde{\QQ}}\left[\mathbbm{1}_{C_i}\left(\int_{s\wedge\tau_N}^{t\wedge\tau_N}(b^i_r+A^i_r\nu_r)dr+\int_{s\wedge\tau_N}^{t\wedge\tau_N}A^i_rd\tilde{B}_r)\right)\right]=\EE^{\tilde{\QQ}}\left[\mathbbm{1}_{C_i}\int_{s\wedge\tau_N}^{t\wedge\tau_N}(b^i_rdr+A^i_rdX_r)\right]  \\
&=\EE^{\WW}\left[\mathbbm{1}_{C_i}\int_{s\wedge\tau_N}^{t\wedge\tau_N}(b^i_rdr+A^i_rdX_r)\bar{M}_T(\nu)\right] \leq 0,
\end{split}
\end{equation*}
where $\tilde{B}$ is a $\RR^k$-valued, $\tilde{\QQ}$-brownian motion. As $A_{.\wedge \tau_N}$ is bounded we observe that the process $\int_{0}^{.\wedge\tau_N}A^i_rd\tilde{B}_r$ is a $\tilde{\QQ}$-martingale. Therefore, 
\begin{equation*}
\EE^{\tilde{\QQ}}\left[\mathbbm{1}_{C_i}\EE^{\tilde{\QQ}}\left(\int_{s\wedge\tau_N}^{t\wedge\tau_N}(b^i_r+A^i_r\nu_r)dr|\mathcal{F}_{s\wedge\tau_N}\right)\right] = \EE^{\tilde{\QQ}}\left[\mathbbm{1}_{C_i}\int_{s\wedge\tau_N}^{t\wedge\tau_N}(b^i_r+A^i_r\nu_r)dr\right] \leq 0,
\end{equation*}
 which implies $\tilde{\QQ}(C_i) = 0$. Hence,
\begin{equation*}
    \EE^{\tilde{\QQ}} \left(\int_{s\wedge\tau_N}^{t\wedge\tau_N}(b^i_r+A^i_{r}\nu_r)dr| \mathcal{F}_{s\wedge\tau_N} \right) \leq 0, \quad \tilde{\QQ}-a.s., \quad i \in \{1,\cdots,l\}. 
\end{equation*}
The latter shows that for every $N \in \NN$, the process $Y^N$ defined by
\begin{equation*}
    Y_t^N := \int_{0}^{ t\wedge\tau_N}(b_r+A_{r}\nu_r)dr,
\end{equation*}
is a $\tilde{\QQ}$-supermartingale with respect to the filtration
$\mathbb{F}^{\tau_N} :=(\mathcal{F}_{t\wedge \tau_N})_{t \in [0,T]}$.
\\
By the Doob-Meyer decomposition (\cite{KaratzasShreve}, Theorem 4.10) the process $Y^N$ is non-increasing. The previous implies that for all $N \in \NN$:
\begin{equation*}
    b_t +A_t\nu_t \leq 0, \quad 0 \leq t \leq  T\wedge\tau_N,\quad  \tilde{\QQ}-a.s.
\end{equation*}
Taking the limit when $N \rightarrow \infty$: 
\begin{equation*}
      b_t +A_t\nu_t \leq 0, \quad 0 \leq t \leq  T,\quad  \tilde{\QQ}-a.s.
\end{equation*}
Using $\tilde{\QQ} \sim \WW$, we obtain:
\begin{equation*}
     b_t +A_t\nu_t \leq 0, \quad 0 \leq t \leq T, \quad  {\WW}-a.s.
\end{equation*}
Then, property d in Definition \ref{def:A} is satisfied. 
Finally, we obtain from Lemma \ref{lemma:condexpM}
\begin{equation*}
\EE^{\WW}\left[\bar{M}_T(\nu)\log\left(\bar{M}_T(\nu)\right) \right] =\frac{1}{2} \EE^{\QQ}\left[M\int_0^T\|\nu_s\|^2ds \right] < \infty. 
\end{equation*}
Hence, $\nu \in \mathcal{A}$. \qed
\end{proof}

\smallskip

Next, we state the following additional assumptions imposing the concavity of $U_{a/p}$ in $M$. 

\begin{ass}\label{ass:2}
Assume that
\begin{itemize}
\item[a)]$U_a(\xi,x,m)$ is strictly concave as a function of $m> 0$ for any given fixed $(\xi,x)$,
\item[b)] or $U_a(\xi,x,m)$ and $U_p(\xi,x,m)$ are concave as functions of $m> 0$ for any given fixed $(\xi,x)$.
\end{itemize}
\end{ass}

\medskip

The above assumptions provide sufficient conditions under which an optimal control of the agent in the relaxed formulation can be constructed in the strong form $\QQ^\nu$, for some $\nu\in\mathcal{A}$ (recall \eqref{eq.hatQ.def}).
    
\begin{lemma}\label{lemma:strong:is:weak}
$\hat{\QQ}^\nu \in \mathcal{U}$ for all $\nu \in \mathcal{A}$.
\end{lemma}
\begin{proof}
    Let $\nu \in \mathcal{A}$. We see that $\hat{\QQ}^\nu$ satisfies property $1$ in Definition \ref{def:U} as a consequence of the properties of $\mathcal{A}$: 
\begin{equation*}
    \EE^{\hat{\QQ}^\nu} M = \EE^{\WW} \left[\bar{M}_T(\nu)\right] = 1. 
\end{equation*}
Furthermore, $\hat{\QQ}^\nu$ satisfies the property 2 in Definition \ref{def:U} as
\begin{equation*}
    \hat{\QQ}^\nu\left(M > 0 \right) = \WW \left(\bar{M}_T(\nu) > 0\right) = 1.
\end{equation*}
Property 3 in Definition \ref{def:U} is trivial to see from \eqref{eq.hatQ.def}.
Finally, we need to show property $4$ in Definition \ref{def:U}. By Girsanov theorem and Lemma \ref{lemma:condexpM}: 
\begin{equation*}
\begin{split}
\EE^{\hat{\QQ}^\nu}\left[M\eta^{\top}\int_{s\wedge\tau_N}^{t\wedge\tau_N} \left( b_r\,dr+ A_{r}dX_r \right) \right] &= \EE^{\WW}\left[\bar{M}_T(\nu)\eta^{\top}\int_{s\wedge\tau_N}^{t\wedge\tau_N} \left( b_r\,dr+ A_{r}dX_r \right) \right]  \\
&= \EE^{\bar{\QQ}^\nu}\left[\eta^{\top} \int_{s\wedge\tau_N}^{t\wedge\tau_N}(b_rdr+A_{r}\nu_r)dr\right],
\end{split}
\end{equation*}
where  $d\bar{\QQ}^\nu = \bar{M}_T(\nu)d\WW$. 

Using that $\eta$, $A_{.\wedge\tau_N}$, and $b_{.\wedge\tau_N}$  are bounded, continuous processes we obtain: 
\begin{equation*} \EE^{\hat{\QQ}^\nu}\left[M\eta^{\top}\int_{s\wedge\tau_N}^{t\wedge\tau_N} \left( b_r\,dr+ A_{r}dX_r \right) \right] = \EE^{\bar{\QQ}^\nu}\left[\eta^{\top} \int_{s\wedge\tau_N}^{t\wedge\tau_N}(b_rdr+A_{r}\nu_rdr)\right] \leq 0,
\end{equation*}

where the last inequality comes from the fact that $\nu \in \mathcal{A}$. Finally, 
using Jensen's inequality:
\begin{equation*}
\EE^{\hat{\QQ}^\nu}\left[\log(M)M\right] = \EE^{\WW}\left[\log(\bar{M}_T(\nu))\bar{M}_T(\nu) \right]   < \infty. 
\end{equation*}
Therefore, $\hat{\QQ}^\nu \in \mathcal{U}$. \qed
\end{proof}

\smallskip

\begin{lemma}\label{prop:relaxed.is.strong}
Under Assumption \ref{ass:2}.a, for any $\xi\in\mathcal{C}^a$ and any $\QQ\in\mathcal{U}^*(\xi)$, there exists $\nu\in\mathcal{A}$ such that $\QQ=\hat{\QQ}^\nu$.
\end{lemma}
\begin{proof}
Consider arbitrary $\xi\in \mathcal{C}^a$ and $\QQ \in \mathcal{U}^*(\xi)$, and define
\begin{equation}\label{def_q_hat}
\hat{\QQ}(dx,dm) := \delta_{\EE^{\QQ}(M|\mathcal{F}_T)}(dm)\WW(dx).
\end{equation}
Firstly, we apply Lemma \ref{le:strong.sol.le1} to deduce the existence of $\nu\in\mathcal{A}$ such that $\EE^{\QQ}(M|\mathcal{F}_T)=\bar M_{T}(\nu)$. The latter implies $\hat\QQ=\hat\QQ^\nu$, and, in turn, Lemma \ref{lemma:strong:is:weak} yields $\hat\QQ^\nu \in \mathcal{U}$.

Next, we recall Assumption \ref{ass:2}.a which states that $U_a(\xi(x),x,.)$ is strictly concave for every $x \in C([0,T],\RR^k)$. Using Jensen's inequality:
\begin{equation*}
    V_a(\xi) = \EE^{\QQ}\left[U_a(\xi(X),X,M)\right] \leq  \EE^{\QQ}\left[U_a\left(\xi,X,\EE^{\QQ}\left(M|\mathcal{F}_T\right)\right)\right] =  \EE^{\hat{\QQ}^{\nu}}\left[U_a(\xi(X),X,M) \right].
\end{equation*}
Hence, $\hat{\QQ}^{\nu} \in \mathcal{U}^*(\xi)$.

Let us assume that $\hat{\mathbb{Q}}^{\nu}\neq \QQ $. Using Assumption \ref{ass:2}.a, for any   $\lambda \in (0,1)$: 
\begin{equation*}
    \begin{split}
       &\lambda U_a(X,\xi(X),M) + (1-\lambda)U_a(X,\xi(X),\EE^{\QQ}\left(M|\mathcal{F}_T\right)) \\
       &< U_a\left(\xi(X),X,\lambda M+ (1-\lambda)\EE^{\QQ}(M|\mathcal{F}_T)\right) \quad \QQ- a.s..
       \end{split}
\end{equation*}
Taking expectations under $\QQ$, using that $\QQ, \hat{\QQ}^{\nu} \in \mathcal{U}^*(\xi)$, and applying Jensen's inequality, we obtain:
\begin{equation*}
    \begin{split}
    V_a(\xi)&< \EE^{\QQ}\left[U_a\left(\xi(X),X,\lambda M + (1-\lambda)\EE^{\QQ}(M|\mathcal{F}_T)\right) \right] \\
    &\leq\EE^{\QQ}\left[U_a\left(\xi(X),X,\lambda \EE^{\QQ}(M|\mathcal{F}_T) + (1-\lambda)\EE^{\QQ}(M|\mathcal{F}_T)\right) \right]  \\
    &=\EE^{\QQ}\left[U_a\left(\xi(X),X,\EE^{\QQ}(M|\mathcal{F}_T)\right) \right] =\EE^{\hat{\QQ}^{\nu}}\left[ U_a\left(\xi(X),X,M\right)  \right] =V_a(\xi).
\end{split}
\end{equation*}
And we obtain a contradiction. Hence, $\QQ =\hat{\QQ}^{\nu} $. \qed
\end{proof}

\medskip

The following proposition summarizes the relationship between solutions of the agent's problem, in the strong and in the relaxed formulations.

\begin{proposition}\label{value_agent_equal}
    Under Assumption  \ref{ass:2}.a or \ref{ass:2}.b we have
    \begin{enumerate}
        \item ${V}_a(\xi) = \bar{V}_a(\xi)$, for any  
 $\xi \in \mathcal{C}$. 
     \item $\mathcal{C}^a = \bar{\mathcal{C}}^a$.
    \item  For any $\xi \in \mathcal{C}$ and $\nu \in \mathcal{A}$, we have: $\nu \in \mathcal{A}^*(\xi) $ if and only if  \smallskip $\hat{\QQ}^\nu \in \mathcal{U}^*(\xi)$.
     \end{enumerate}
\end{proposition}
\begin{proof}
   Let us show the first statement. For any $\xi \in \mathcal{C}$ and $\nu\in\mathcal{A}$, we have:
\begin{equation}\label{aux_stament2}
      \EE^{\WW}\left[U_a(\xi(X),X,\bar{M}_T(\nu))\right] = \EE^{\hat{\QQ}^\nu}\left[U_a(\xi(X),X,M) \right].
   \end{equation}
   By Lemma \ref{lemma:strong:is:weak} we have that $\hat{\QQ}^\nu \in \mathcal{U}$. 
 Taking the supremum in \eqref{aux_stament2} we obtain: 
   \begin{equation*}
   \begin{split}
       \bar{V}_a(\xi) &= \sup_{\nu \in \mathcal{A}} \EE^{\WW}\left[U_a(\xi(X),X,\bar{M}_T(\nu))\right] = \sup_{\nu \in \mathcal{A}} \EE^{\hat{\QQ}^\nu}\left[U_a(\xi(X),X,M) \right] \\
       &\leq \sup_{\QQ \in \mathcal{U}} \EE^{\QQ}\left[U_a(\xi(X),X,M) \right] = V_a(\xi).               
   \end{split}
       \end{equation*}
      
      Using the concavity of $U_a(\xi,x,.)$, we obtain for any $\QQ \in \mathcal{U}$:
     \begin{equation}\label{eq.aux.statement3}
    \EE^{\QQ}\left[U_a(\xi(X),X,M)\right] \leq \EE^{\WW}\left[U_a\left(\xi(X),X,\EE^{\QQ}(M|\mathcal{F}_T)\right)\right].
    \end{equation}
    By Lemma \ref{le:strong.sol.le1} we have that for any $\QQ \in \mathcal{U}$ there exists $\nu \in \mathcal{A}$ such that $\EE^{\QQ}\left(M|\mathcal{F}_T \right)= \bar{M}_T(\nu)$. Hence, taking the supremum in equation \eqref{eq.aux.statement3}, and using Lemma \ref{le:strong.sol.le1}:
\begin{equation*}   
    \begin{split}
        &V_a(\xi) = \sup_{\QQ\in \mathcal{U}}\EE^{\QQ}\left[U_a\left(\xi(X),X,M\right)\right] \leq \sup_{\QQ\in \mathcal{U}}\EE^{\WW}\left[U_a\left(\xi(X),X,\EE^{\QQ}(M|\mathcal{F}^X_T)\right)\right] \\
        &\leq \sup_{\nu\in \mathcal{A}}\EE^{\WW}\left[U_a\left(\xi(X),X,\bar{M}_T(\nu)\right)\right]
         =\bar{V}_a(\xi).
    \end{split}
 \end{equation*}
Therefore, $V_a(\xi) = \bar{V}_a(\xi)$.

\smallskip

Let us show the second statement of the proposition. Consider any $\xi \in \mathcal{C}^a$ and $\QQ^* \in \mathcal{U}^*(\xi)$. By Lemma \ref{le:strong.sol.le1}, there exits $\nu^* \in \mathcal{A}$ such that $\EE^{\QQ^*}\left(M|\mathcal{F}_T\right)=\bar{M}_T(\nu^*)$. Hence, using the concavity of $U_a(\xi,x,.)$ and the first statement of the proposition, we have: 
\begin{equation*}
    \begin{split}
    & \bar{V}_a(\xi) = V_a(\xi) = \EE^{\QQ^*}\left[ U_a(\xi(X),X,M)\right] \leq \EE^{\QQ^*}\left[ U_a(\xi(X),X,\EE^{\QQ^*}(M|\mathcal{F}_T))\right] \\
    &= \EE^{\WW}\left[ U_a(\xi(X),X,\bar{M}_T(\nu^*))\right].
    \end{split}
\end{equation*}
Hence, 
$\nu^* \in \mathcal{A}^*(\xi)$, which implies that $\xi \in \mathcal{C}^a$. 

Next, consider any $\xi \in \bar{C}^a$. Then, there exits $\nu^* \in \mathcal{A}$ such that 
\begin{equation}    
\EE^{\WW}\left[ U_a\left(\xi(X),X,\bar{M}_T(\nu^*)\right) \right] = \bar{V}_a(\xi).
\end{equation}
Using Lemma \ref{lemma:strong:is:weak}, we have that $\hat{\QQ}^{\nu^*}\in \mathcal{U}$. Hence, by the first statement of the proposition:
\begin{equation*}
    \EE^{\hat{\QQ}^{\nu^*}}\left[U_a\left(\xi(X),X,M\right) \right] = \EE^{\WW}\left[U_a\left(\xi(X),X,\bar{M}_T(\nu^*)\right)\right]  =\bar{V}_a(\xi) = V_a(\xi), 
\end{equation*}
which yields  $\hat{\QQ}^{\nu^*} \in \mathcal{U}^*(\xi)$ and, in turn, implies that $\xi \in \mathcal{
C}^a$. Thus, $\mathcal{C}^a= \bar{\mathcal{C}}^a$.

\smallskip

The third statement of the proposition is deduced trivially from the first one. Indeed, for any $\nu \in \mathcal{A}^*(\xi)$, 
\begin{equation}\label{aux_statement_prop}
    V_a(\xi) = \bar{V}_a(\xi) = \EE^{\WW}\left[U_a(\xi(X),X,\bar{M}_T(\nu))\right] = \EE^{\hat{\QQ}^\nu}\left[U_a(\xi(X),X,M)\right].
\end{equation}
Hence, $\hat{\QQ}^\nu \in \mathcal{U}^*(\xi)$. Conversely, if $\hat{\QQ}^\nu \in \mathcal{U}^*(\xi)$, we obtain from \eqref{aux_statement_prop} that $\nu \in \mathcal{A}^*(\xi)$. 
 \qed
 \end{proof}
 
\medskip 
 
The following theorem summarizes the results of this section, showing that the control problems of the principal are the same, in the strong and in the relaxed formulations.
 
\begin{theorem}\label{cor:main}
Under Assumption \ref{ass:2}.a or \ref{ass:2}.b, we have: $\mathcal{C}^a = \bar{\mathcal{C}}^a$ and $\bar{J}_p(\xi)= J_p(\xi)$ for any $\xi\in\mathcal{C}^a$.
\end{theorem}          
\begin{proof}
The equality $\mathcal{C}^a = \bar{\mathcal{C}}^a$ is proven in Proposition \ref{value_agent_equal} (and is only included in this theorem for completeness). Thus, we only prove the second statement.

Consider any $\xi \in \mathcal{C}^a$ and notice that $\xi \in \bar{\mathcal{C}}^a$. 
If Assumption \ref{ass:2}.a  holds, we have from Lemma \ref{prop:relaxed.is.strong} that for every $\QQ \in \mathcal{U}^*(\xi)\neq \emptyset$, there exits $\nu \in \mathcal{A}$ such that $\hat{\QQ}^\nu = \mathcal{\QQ}$. Additionally, from Proposition \ref{value_agent_equal} we have that $\nu \in \mathcal{A}^*(\xi)$. Proposition \ref{value_agent_equal} also shows that for any $\nu \in \mathcal{A}^*(\xi)$ we have $\hat{\QQ}^\nu\in \mathcal{U}^*(\xi)$. Therefore, 
\begin{equation*}
    \begin{split}
    &J_p(\xi) = \sup_{\QQ \in \mathcal{U}^*(\xi)}\EE^{\QQ}\left[U_p(\xi(X),X,M) \right] = \sup_{\nu \in \mathcal{A}^*(\xi)}\EE^{\hat{\QQ}^\nu}\left[U_p(\xi(X),X,M) \right] \\
    &= \sup_{\nu \in \mathcal{A}^*(\xi)}\EE^{\WW}\left[U_p(\xi(X),X,\bar{M}_T(\nu)) \right] = \bar{J}_p(\xi).
    \end{split}
\end{equation*} 

\smallskip

If Assumption \ref{ass:2}.b holds, we have, for every $\QQ \in \mathcal{U}^*(\xi)$:
\begin{align}
&\EE^{\QQ}\left[U_p(\xi(X),X,M) \right] \leq  \EE^{\QQ}\left[U_p\left(\xi(X),X,\EE^{\QQ}(M|\mathcal{F}_T)\right) \right]\nonumber\\
&=\EE^{\hat{\QQ}^\nu}\left[ U_p(\xi(X),X,M) \right] = \EE^{\WW}\left[U_p(\xi(X),X,\bar{M}_T(\nu)) \right],\label{aux_statement1_theorem2}
\end{align}
\begin{align}
&V_a(\xi) = \EE^{\QQ}\left[U_a(\xi(X),X,M) \right] \leq  \EE^{\QQ}\left[U_a\left(\xi(X),X,\EE^{\QQ}(M|\mathcal{F}_T)\right) \right]\nonumber\\
&= \EE^{\hat\QQ^{\nu}}\left[U_a\left(\xi(X),X,M\right) \right] = \EE^{\WW}\left[U_a(\xi(X),X,\bar{M}_T(\nu))\right],\label{aux_statement1_theorem2.2}
\end{align}
where $\nu \in \mathcal{A}$ satisfying $\EE^{\QQ}\left(M|\mathcal{F}_T \right) = \bar{M}_T(\nu)$ is given by Lemma \ref{le:strong.sol.le1}. 

From equation \eqref{aux_statement1_theorem2.2} we observe that $\hat{\QQ}^\nu \in \mathcal{U}^*(\xi)$, for any $\QQ \in \mathcal{U}^*(\xi)$ and $\nu$ associated with $\QQ$ as in the above. Moreover,
by Proposition \ref{value_agent_equal}, we have that $\nu \in \mathcal{A}^*(\xi)$. Therefore, taking supremum in \eqref{aux_statement1_theorem2}, we obtain: 
\begin{equation*}
\begin{split}
J_p(\xi) &= \sup_{\QQ \in \mathcal{U}^*(\xi)} \EE^\QQ \left[U_p(\xi,X,M) \right] \leq \sup_{\nu \in \mathcal{A}^*(\xi)} \EE^{\hat{\QQ}^\nu}\left[U_p(\xi,X,M) \right] \\
&= \sup_{\nu \in \mathcal{A}^*(\xi)} \EE^{\WW} \left[U_p(\xi,X,\bar{M}_T(\nu))\right] = \bar{J}_p(\xi).
\end{split}
\end{equation*}
On the other hand, using Proposition \ref{value_agent_equal}, we have $\hat{\QQ}^\nu \in \mathcal{U}^*(\xi)$ for any $\nu \in \mathcal{A}^*(\xi)$. Hence, 
\begin{equation*}    
\begin{split}
\bar{J}_p(\xi) &= \sup_{\nu \in \mathcal{A}^*(\xi)}\EE^{\WW} \left[U_p(\xi(X),X,\bar{M}_T(\nu))\right] =\sup_{\nu \in \mathcal{A}^*(\xi)}\EE^{\hat{\QQ}^\nu} \left[ U_p\left(\xi(X),X,M\right)\right] \\
&\leq \sup_{\QQ \in \mathcal{U}^*(\xi)}\EE^{{\QQ}} \left[U_p\left(\xi(X),X,M\right)\right] = J_p(\xi). 
\end{split}
\end{equation*}
Thus, we obtain that $J_p(\xi) = J_p(\xi)$ for all $\xi \in \mathcal{C}^a$. 
\qed
\end{proof} 

\smallskip

The above theorem, in particular, implies that the solution $\xi^*$ of the relaxed problem, constructed in the preceding section, is also a solution to the strong optimal contract problem, under Assumptions \ref{ass:1}, \ref{ass:1.5}, \ref{ass:2}.a,  \ref{ass:2}.b.

\section{Optimal brokerage fee for a client with a private trading signal}
\label{opt:brok:as:info}

Herein, we apply the results obtained in previous sections to the problem of brokerage fees with information asymmetry. 
We introduce $\Omega:= C([0,T],\RR^{3})$ equipped with the Wiener measure $\WW$. We denote by $X := (P,Z,W)^\top$ the canonical random element on $\Omega$, define the measure $\WW^{\sigma,\epsilon}:=\WW \circ (\sigma P,\epsilon Z, W)^{-1}$, and consider the completed (under $\WW^{\sigma,\epsilon}$) filtration $\FF:= \FF^X$.
We denote by $\mathcal{D}$ the set of agent's actions, which consists of all $\FF$-progressively measurable process $\pi$ bounded from below and from above, respectively, by the constants $L$ and $U$. Given an action $\pi \in \mathcal{D}$, we introduce the measure $\QQ^{\pi}\sim\mathbb{W}^{\sigma,\epsilon}$ under which the canonical process $X=(P,Z,W)^\top$ satisfies:
\begin{align*}
&dP_t = W_tdt + \sigma d\tilde{B}^{\pi}_t,\\
&dZ_t = \pi_t dt + \epsilon d\hat{B}^{\pi}_t, 
\end{align*}
where $\tilde B^{\pi},\hat B^{\pi}$ are independent  $\QQ^\pi$-brownian motions.
The coefficients satisfy:  $\sigma,\epsilon>0$.
To ensure that $\QQ^{\pi}$ is well dfined, we need that
\begin{align*}
\hat M(\pi):=\frac{d\QQ^{\pi}}{d\mathbb{W}^{\sigma,\epsilon}}&=\exp\left(-\frac{1}{2}\int_0^T\left(\frac{1}{\sigma^2}W^2_t + \frac{1}{\epsilon^2}\pi^2_t\right)\,dt + \frac{1}{\sigma^2}\int_0^T W_t dP_t + \frac{1}{\epsilon^2}\int_0^T \pi_t\,dZ_t\right)
\end{align*}
satisfies $\EE^{\mathbb{W}}\hat M(\pi)=1$. The above condition is satisfied for all $\pi \in \mathcal{D}$ as a consequence of Corollary 5.16 in \cite{KaratzasShreve}. 

The process $\pi$ represents the trading rate of the agent, playing the role of his control. We assume there is a minimum and maximum attainable rates $L$ and $U$.  
The process $P$ represents the price of the risky asset (the riskless asset has zero return).
The process $Z$ represents the inventory of the agent perturbed by $\epsilon \hat{B}^{\pi}$. The latter perturbation is interpreted as the uncontrolled changes in the agent's inventory, which may be due to rounding errors, automatic portfolio adjustments (e.g., to preserve diversification or leverage ratio), hedging needs, internal order flow (e.g., if the agent is an execution desk in larger firm), etc.

\medskip

Using the notation
\begin{align}
&X=(P,Z,W)^\top,\quad \nu:=(W,\pi,0)^\top,\nonumber \quad B^{\pi} := (\sigma\tilde{B}^\pi,\epsilon\hat{B}^\pi,W)^\top, \\
&b := \left(
\begin{array}{c}
{-W}\\
{W}\\
{0}\\
{0}\\
{-U}\\
{L}
\end{array}
\right),
\quad A := \left(
\begin{array}{ccc}
{1} & {0} & {0}\\
{-1} & {0} & {0}\\
{0} & {0} & {1}\\
{0} & {0} & {-1}\\
{0} & {1} & {0}\\
{0} & {-1} & {0}
\end{array}
\right),
\quad \bar M_T(\nu):=\hat M(\pi),\label{eq.example.A.b.def}
\end{align}
where $L\leq U$ are arbitrary constants, we embed this model into the setting described by Definition \ref{def:A} and equation \eqref{eqn.state}.

\medskip

Next, we consider a set of contracts $\mathcal{C}$ which can be any non-empty set of continuous mappings $\xi$ from $C([0,T],\RR^2)$ to $\RR$, which is compact with respect to the topology of uniform convergence on compacts and is such that
\begin{equation}\label{ui.contracts}
    \lim_{N\rightarrow \infty} \sup_{(\xi,\pi) \in \mathcal{C}\times\mathcal{D}}\EE^{\QQ^{\pi}}\left[|\xi(X)| \mathbbm{1}_{|\xi(X)|\geq N}\right] = 0. 
\end{equation}
We also require that $\mathcal{C}$ contains a large enough constant function (this is clarified at the end of this section).
A sufficient condition for \eqref{ui.contracts} is given, e.g., by the following: for any $r>0$, the family of random variables
\begin{align}
\left\{\sup_{y\in C([0,T],\RR^3):\,\|y\|_{\mathcal{C}([0,T],\RR^3)}\leq 1}|\xi(X+y\,r\,(1+\|X\|))|\right\}_{\xi\in\mathcal{C}}
\label{ui.contracts.2}
\end{align}
is uniformly integrable under $\mathbb{W}^{\sigma,\epsilon}$.

Let us describe examples of sets $\mathcal{C}$ which satisfy the above assumptions. First, for $\gamma, M, K \in (0,\infty)$, we define
\begin{equation*}
    \mathcal{C}_1^{\gamma, M,K} := \left\{\xi : C([0,T],\RR^2) \mapsto [-K,K]: \left|\xi(X)-\xi(Y)\right| \leq M\|X-Y\|^\gamma_{C([0,T],\RR^2)}, \, \forall X,Y\in C([0,T],\RR^2) \right\}.
\end{equation*} 
It is easy to see that $\mathcal{C}_1^{\gamma, M,K}$ satisfies \eqref{ui.contracts.2} and the aforementioned compactness property by the Ascoli's theorem (see \cite[Theorem 18]{Kelley}).
In addition, any closed subset of $\mathcal{C}_1^{\gamma, M,K}$ also satisfies the above assumptions: this, in particular includes the set of all functions in $\mathcal{C}_1^{\gamma, M,K}$ that depend on the values of $(P,Z)$ only at a given finite partition of $[0,T]$, etc.
Another example arises if we consider contracts of linear-polynomial type. Namely, we consider $K,n>0$ and a linear operator $\mathcal{L}$, bounded in $C([0,T],\RR)$, and define 
\begin{equation*}
    \mathcal{C}_2^{K,n}
     := \left\{\xi:\, \xi(P,Z)=\sum_{i,j=1}^n a_{ij} (\mathcal{L}(P))^i (\mathcal{L}(Z))^j,\quad a_{ij} \in [-K,K] \right\}.
\end{equation*}
The operator $\mathcal{L}$ may be a integral operator, or evaluation at a specific point, etc. It is easy to see that $\mathcal{C}_2^{K,n}$ satisfies \eqref{ui.contracts.2} and the aforementioned compactness property.


\begin{remark}
Note that in the above examples we restricted the set of admissible contracts $\xi$ to only those that depend on the first two components of the state process, $(P,Z)$. This represents the asymmetry of information described in the introduction. Indeed, the trading signal $W$ is a private signal of the client, hence the contract is only allowed to depend on $(P,Z)$ which are observed by the broker.
\end{remark}

\medskip

Let us describe the objectives of the principal (broker) and of the agent (client).
The agent's profits are
\begin{align*}
\int_0^T Z_t dP_t = \int_0^T Z_t W_t dt + \sigma \int_0^T Z_t d\tilde B^\pi_t.
\end{align*}
For a given contract $\xi$, the agent maximizes over all admissible controls $\pi \in \mathcal{D}$ his expected profit less the brokerage fee and a quadratic penalty for high trading rate:
\begin{align*}
&\EE^{\QQ^{\pi}}\left[-\xi(X) + \int_0^T Z_t W_t dt + \sigma \int_0^T Z_t d\tilde B^\pi_t - \phi_a\int_0^T \pi^2_tdt \right]\\
&=\EE^{\QQ^{\pi}}\left[-\xi(X) + \int_0^T Z_t W_t dt - \phi_a\int_0^T \pi^2_tdt \right]\\
&=\EE^{\QQ^{\pi}}\left[-\xi(X) + \int_0^T Z_t W_t dt - 2\epsilon^2\phi_a\log \hat M(\pi) + \frac{\epsilon^2\phi_a}{\sigma^2}\int_0^T W^2_t dt \right]\\
&= \EE^{\mathbb{W}^{\sigma,\epsilon}}\left[\hat M(\pi)\left(-\xi(X) + \int_0^T Z_t W_t dt - 2\epsilon^2\phi_a\log \hat M(\pi) + \frac{\epsilon^2\phi_a}{\sigma^2}\int_0^T W^2_t dt\right)\,\right]\\
&= \EE^{\mathbb{W}^{\sigma,\epsilon}}\left[-\hat M(\pi)\,\xi(X) - 2\epsilon^2\phi_a \hat M(\pi)\,\log \hat M(\pi) + \hat M(\pi)\int_0^T\left(\frac{\epsilon^2\phi_a}{\sigma^2} W^2_t+Z_tW_t\right)\,dt\right],
\end{align*}
where $\phi_a>0$ is a given constant.
Thus, the agent's objective is
\begin{align*}
\bar{J}_a(\nu,\xi)&:= \EE^{\mathbb{W}^ {\sigma,\epsilon}}\left[-\bar M_T(\nu)\,\xi(X) -2\epsilon^2\phi_a \bar M_T(\nu)\,\log \bar M_T(\nu) +\bar M_T(\nu) \int_0^T\left(\frac{\epsilon^2\phi_a}{\sigma^2}W_t^2+Z_tW_t \right)\,dt\right] \\
&= \EE^{\mathbb{W}^{\sigma,\epsilon}}\left[U_a\left(\xi(X),X,\bar{M}_T(\nu)\right)\right],\\
U_a\left(\xi,X,M\right) &:=-M\xi -2\epsilon^2\phi_a M\log M + M\int_0^T\left(\frac{\epsilon^2\phi_a}{\sigma^2}W_t^2+Z_tW_t \right)\,dt.
\end{align*}

\smallskip

The broker maximizes over all fees $\xi \in \bar{\mathcal{C}}^a$ her expected fee less the penalty for high inventory: 
\begin{align*}
\bar{J}_p(\xi)&:=\sup_{\nu\in\mathcal{A}^*(\xi)}\EE^{\QQ^{\pi}}\left[\xi(X)-\phi_p \int_0^T (\pi_s)^2ds\right] \\
&= \sup_{\nu\in\mathcal{A}^*(\xi)}\EE^{\WW^{\sigma,\epsilon}}\left[\bar{M}_T(\nu) \xi(X) -2\epsilon^2\phi_p\bar{M}_T(\nu) \log \bar{M}_T(\nu)+\bar{M}_T(\nu)\int_0^T\frac{\epsilon^2\phi_p}{\sigma^2}W_t^2\,dt\right] \\
&=\sup_{\nu\in\mathcal{A}^*(\xi)}\EE^{\WW^{\sigma,\epsilon}}\left[U_p\left(\xi(X),X,\bar{M}_T(\nu) \right) \right],\\
U_p\left(\xi,X,M\right) &:=M \xi-2\epsilon^2\phi_pM \log M+M\int_0^T\frac{\epsilon^2\phi_p}{\sigma^2}W_t^2\,dt.
\end{align*}
We recall that $\mathcal{A}^*(\xi)$ is the set of optimal controls of the agent given the contract $\xi$, and that $R_a$ denotes the reservation value of the agent. 


\medskip

Let us now verify that the assumptions made in previous sections are satisfied for the problem at hand. Noticing that the function $x\mapsto -x\log x$ is strictly concave, we deduce that Assumption \ref{ass:2}.a is satisfied. Then, we conclude (see Theorem \ref{cor:main}) that the above strong formulation of an optimal contract problem has the same solution as its relaxed version. It remains to show that there exists an optimal contract in the relaxed formulation \eqref{eq.agentProblem.def.3}--\eqref{eq.principalProblem.def.relaxed}, with $\mathcal{U}$ given by Definition \ref{def:U} and with $A$, $b$ defined earlier in this section. To this end we notice that the functions $U_a, U_p$, defined above, are continuous in $C([0,T],\RR^3)\times (0,\infty)$. Hence Assumption \ref{ass:1}.a is satisfied. Next, for any $\QQ\in\mathcal{U}$, we apply Lemma \ref{le:strong.sol.le1} to deduce the existence of $\pi\in\mathcal{D}$ such that $\EE^{\QQ} \left(M\,\vert\,\mathcal{F}_T\right) = \hat{M}(\pi)$. Denoting
\begin{align*}
& \zeta:=  \int_0^T\left(\frac{\epsilon^2\phi_a}{\sigma^2}W_t^2+Z_tW_t \right)dt,
\end{align*}
we obtain, for any $\QQ\in\mathcal{U}$ and any $\xi\in\mathcal{C}$:
\begin{align*}
& \EE^{\QQ} \left[U_{a}(\xi(X),X,M)\right]^+ \leq 2\epsilon^2\phi_{a}\sup_{m>0} [-m\log m] + \EE^{\QQ^\pi} (\zeta-\xi)^+ < \infty.
\end{align*}
Similarly, we show that the expectation of the positive part of principal's utility is finite. Thus, we verify the first pat of Assumption \ref{ass:1}.b (the second part of Assumption \ref{ass:1}.b is verified at the end of this section).

Next, we verify Assumption \ref{ass:1.5} with 
\begin{align*}
 \mathcal{U}_0:=\{\delta_{\hat M(\pi(p,z,w))}(dm)\mathbb{W}(dp,dz,dw)\,\}_{\pi\in\mathcal{D}}.
\end{align*}
To this end, we repeat the arguments in the proof of Lemma \ref{prop:relaxed.is.strong} to obtain, for all $\xi \in \mathcal{C}$:
\begin{align*}
&\mathcal{U}^*(\xi) \subset \mathcal{U}_0, \quad
\sup_{\QQ \in \mathcal{U}}\EE^{\QQ}\left[U_a(\xi(X),X,M)\right] =  \sup_{\QQ \in \mathcal{U}_0}\EE^{\QQ}\left[U_a(\xi(X),X,M)\right] =  \sup_{\pi \in \mathcal{D}}\EE^{\QQ^\pi}\left[U_a(\xi(X),X,\hat M(\pi))\right].
\end{align*}
Next, we notice that
\begin{align*}
& U_a(\xi,X,M) =  M\left(-\xi -2\epsilon^2 \phi_a \log M + \zeta\right),
\end{align*}
and, for all large enough $N>0$ and for any $\pi\in\mathcal{D}$, we obtain:
\begin{align*}
& \EE^{\mathbb{Q}^{\pi}}\left[ \left|U_a(\xi(X),X,\hat{M}(\pi))\right|\,\mathbbm{1}_{|U_a(\xi(X),X,\hat{M}(\pi)|\geq N}\right] \\
&=\EE^{\QQ^\pi}\left[ |-\xi(X) -2\epsilon^2\phi_a  \log \hat M(\pi) + \zeta|\,\mathbbm{1}_{|U_a(\xi(X),X,\hat{M}(\pi)|\geq N}\right]\\
&\leq \EE^{\QQ^\pi}\left[\left( |\xi(X)| + 2\epsilon^2\phi_a  |\log \hat M(\pi)| + |\zeta|\right)\,\mathbbm{1}_{|U_a(\xi(X),X, \hat{M}(\pi))|\geq N}\right] \\
&= \EE^{\QQ^\pi}\left[\left( |\xi(X)| + 2\epsilon^2\phi_a   |\log \hat M(\pi)| + |\zeta|\right)\,\mathbbm{1}_{\left(|\xi(X)|+2\epsilon^2\phi_a  |\log(\hat{M}(\pi)|+|\zeta| \right)|\hat{M}(\pi)|\geq N}\right] \\
&\leq \EE^{\QQ^\pi}\left[\left( |\xi(X)| +  2\epsilon^2\phi_a   |\log \hat M(\pi)| + |\zeta|\right)\,\mathbbm{1}_{|\xi(X)|^2\geq N/9}\right] \\
&+\EE^{\QQ^\pi}\left[\left( |\xi(X)| 
+  2\epsilon^2\phi_a   |\log \hat M(\pi)| + |\zeta|\right)\,\mathbbm{1}_{4\epsilon^4\phi_a^2  |\log(\hat{M}(\pi))|^2\geq N/9}\right]  \\
&+\EE^{\QQ^\pi}\left[\left( |\xi(X)| 
+ 2\epsilon^2\phi_a   |\log \hat M(\pi)| + |\zeta|\right)\,\mathbbm{1}_{|\zeta|^2\geq N/9}\right] \\
&+\EE^{\QQ^\pi}\left[\left( |\xi(X)| 
+ 2\epsilon^2\phi_a    |\log \hat M(\pi)| + |\zeta|\right)\,\mathbbm{1}_{\hat M (\pi)^2\geq N}\right]
\end{align*}
Recall that $\mathcal{C}$ is a family of uniformly integrable random variables under $\QQ^\pi$, uniformly over all $\pi\in\mathcal{D}$, which yields
\begin{align}\label{eq.example.xi.unif.int}
& \limsup_{N\rightarrow\infty}\sup_{\pi\in\mathcal{D},\xi\in\mathcal{C}}\QQ^{\pi}\left(|\xi(X)|\geq \frac{\sqrt{N}}{3}\right)=0.
\end{align}
It is also easy to see (due to the uniform boundedness of $\pi$) that
\begin{align}\label{eq.example.zeta.unif.int}
& \limsup_{N\rightarrow\infty}\sup_{\pi\in\mathcal{D}}\QQ^{\pi}\left(|\zeta|\geq \frac{\sqrt{N}}{3}\right)=0.
\end{align}
In addition, by Markov's inequality: 
\begin{equation}\label{eq.bound.m}
    \begin{split}
\lim_{N\rightarrow\infty}\sup_{\pi \in \mathcal{D}}\QQ^{\pi}\left(\hat{M}(\pi)\geq \sqrt{N} \right) &= \lim_{N\rightarrow\infty}\sup_{\pi \in \mathcal{D}}\QQ^{\pi}\left(\log\left(\hat{M}(\pi)\right)\geq \log(\sqrt{N}) \right) \\
&\leq \lim_{N\rightarrow\infty} \frac{2}{\log\left(N \right)}\sup_{\pi \in \mathcal{D}}\EE^{\QQ^\pi}\left[\left|\log \left( \hat{M}(\pi)\right)\right| \right].
    \end{split}
\end{equation}
Now, we show that the second moments of $\log(\hat{M}(\pi))$ under $\QQ^\pi$ 
are bounded uniformly over $\pi \in \mathcal{D}$. Indeed, 
\begin{equation}\label{bound:log:m}
    \begin{split}
     &\sup_{\pi\in \mathcal{D}} \EE^{\QQ^\pi} \left[ \left(\log\hat{M}(\pi)\right)^2\right] \\
     &=  \sup_{\pi\in \mathcal{D}} \EE^{\QQ^\pi}\left[\left(\int_0^T\frac{1}{2}\left(\frac{1}{\epsilon^2}\pi_s^2+\frac{1}{\sigma^2}W_s^2  \right)\,ds +\frac{1}{\epsilon}\int_0^T\pi_sd\hat{B}^\pi_s+\frac{1}{\sigma}\int_0^TW_sd\tilde{B}^\pi_s  \right)^2\right] \\
     &\leq \sup_{\pi \in \mathcal{D}}\EE^{\QQ^\pi}\left[3\left(\int_0^T\frac{1}{2}\left(\frac{1}{\epsilon^2}\pi_s^2+\frac{1}{\sigma^2}W_s^2 \right)ds\right)^2 +3\left(\int_0^T\frac{1}{\epsilon}\pi_sd\hat{B}^\pi_s\right)^2+3\left(\frac{1}{\sigma}\int_0^TW_sd\tilde{B}^\pi_s  \right)^2\right] \\
     &\leq \frac{3}{2}\EE^{\WW^{\sigma,\epsilon}}\left[\int_0^T\left(\frac{1}{\epsilon^2}(|L|+|U|)^2+\frac{1}{\sigma^2}W_s^2ds \right)^2 \right] + \frac{3 (|L|+|U|)^2T}{\epsilon^2} + \frac{3T}{\sigma^2} <\infty. 
     \end{split}
\end{equation}
Combining the above expression with \eqref{eq.bound.m}, we obtain 
\begin{equation}\label{eq.bound.m.2}
\lim_{N\rightarrow\infty}\sup_{\pi \in \mathcal{D}}\QQ^{\pi}\left(\hat{M}(\pi)\geq \sqrt{N} \right) = 0. 
\end{equation}
Similarly, we show that the second moments of $\zeta$ under $\QQ^\pi$ are bounded uniformly over $\pi\in\mathcal{D}$, which yields its uniform integrability over all $\QQ^\pi$, for $\pi\in\mathcal{D}$. Thus, collecting \eqref{eq.example.xi.unif.int}, \eqref{eq.example.zeta.unif.int}, \eqref{eq.bound.m.2}, we deduce the uniform integrability of $\xi$, $\log \hat M(\pi)$, $\zeta$ and conclude that
\begin{align*}
&\limsup_{N\rightarrow\infty}\sup_{\pi\in\mathcal{D},\xi\in\mathcal{C}}\EE^{\QQ^\pi}\left[ |U_a(\xi(X),X,\hat M(\pi))|\,\mathbbm{1}_{\left|U_a(\xi(X),X,\hat M(\pi))\right|\geq N}\right] = 0,
\end{align*}
which verifies Assumption \ref{ass:1.5}.$a$. Assumption \ref{ass:1.5}.$b$ is verified similarly.

Next, using \eqref{bound:log:m}, we obtain
\begin{equation*}
    \begin{split}
    \sup_{\QQ \in \mathcal{U}_0} \EE^{\QQ} \left[M \log M\right] = \sup_{\pi\in \mathcal{D}} \EE^{\QQ^\pi} \log\hat{M}(\pi) \leq \sup_{\pi\in \mathcal{D}} \EE^{\QQ^\pi} \left|\log\hat{M}(\pi)\right| < \infty. 
    \end{split}
\end{equation*}
which verifies Assumption \ref{ass:1.5}.d.

Let us now show that $\mathcal{U}_0$ satisfies Assumption \ref{ass:1.5}.c. By Markov's inequality:
\begin{equation}\label{limit_Q:0}
    \begin{split}
    \lim_{k\rightarrow \infty}\sup_{\QQ \in \mathcal{U}_0} \QQ\left(M \leq \frac{1}{k} \right) 
    &= \lim_{k\rightarrow \infty}\sup_{\pi \in \mathcal{D}} \WW^{\sigma,\epsilon} \left(\hat{M}(\pi)^{-1}\geq k \right) \\
    &\leq \lim_{k\rightarrow \infty}\sup_{\pi \in \mathcal{D}} \WW^{\sigma,\epsilon} \left(|\log(\hat{M}(\pi)^{-1})|\geq \log(k) \right) \\
    &\leq \lim_{k\rightarrow \infty} \frac{1}{\log(k)}\sup_{\pi \in \mathcal{D}}\EE^{\WW^{\sigma,\epsilon}}\left|\log\left(\hat{M}(\pi)^{-1}\right)\right|.
    \end{split}
\end{equation}
Using the above and \eqref{bound:log:m}, we obtain:
\begin{equation*}
     \lim_{k\rightarrow \infty}\sup_{\QQ \in \mathcal{U}_0} \QQ\left(M \leq \frac{1}{k} \right) = 0, 
\end{equation*}
which verifies Assumption \ref{ass:1.5}.c.

It remains to verify the second part of Assumption \ref{ass:1}.b: i.e., to show that there exists $\xi\in\mathcal{C}$ and $\QQ\in\mathcal{U}^*(\xi)$ such that $J_a(\QQ,\xi)\geq R_a$.
To this end, we consider the following optimization problem:
\begin{align}
  &\tilde{V} := \sup_{\QQ \in \mathcal{U}} \EE^{\QQ}\left[-2\epsilon^2\phi_a  M\,\log M +M \zeta\right]
    = \sup_{\QQ \in \mathcal{U}_0} \EE^{\QQ}\left[-2\epsilon^2\phi_a  M\,\log M +M \zeta\right],
    \label{eq.example.tildeV}
\end{align}
where the last equality follows from the concavity of the integrand in $M$.
Using the uniform boundedness of $\pi\in\mathcal{D}$, we obtain
\begin{align*}
& \lim_{N\rightarrow\infty} \sup_{\QQ \in \mathcal{U}_0} \EE^{\QQ}\left[M(\zeta-N)^+\right]
= \lim_{N\rightarrow\infty} \sup_{\pi \in \mathcal{D}} \EE^{\QQ^\pi}(\zeta-N)^+ = 0.
\end{align*}
Thus, for any $\varepsilon>0$, there exists $N_0$ s.t., for all $N\geq N_0$,
\begin{align*}
&\sup_{\QQ \in \mathcal{U}_0} \EE^{\QQ}\left[-2\epsilon^2\phi_a  M\,\log M +M \zeta\right]\geq\sup_{\QQ \in \mathcal{U}_0} \EE^{\QQ}\left[-2\epsilon^2\phi_a  M\,\log M + M (\zeta-N)^-\right]
\geq \tilde V - \varepsilon.
\end{align*}
Consider a maximizing sequence $\{\QQ_n\}$ of the left hand side of the above. Using the fact that $\mathcal{U}_0$ satisfies Assumption \ref{ass:1.5} (as established above), we repeat the arguments in the proof of Theorem \ref{existence:weak:opt}, to deduce that $\{\QQ_n\}$ has a limit point $\QQ^*\in\mathcal{U}$. Then, using Skorokhod's representation and Fatou's lemma (see the proof of Theorem \ref{existence:weak:opt}), we obtain
\begin{align*}
&\EE^{\QQ^*}\left[-2\epsilon^2\phi_a  M\,\log M +M \zeta\right] \geq \tilde V - \varepsilon.
\end{align*}
As the above holds for any $\varepsilon>0$, we conclude that $\QQ^*$ is a solution to \eqref{eq.example.tildeV}.
Then, choosing a constant contract $\tilde{\xi}\in\RR$ satisfying $\tilde{\xi}\geq\tilde{V}-R_a$, where $R_a$ is the reservation value of the agent, and assuming that $\tilde\xi\in\mathcal{C}$ (i.e., assuming that $\mathcal{C}$ contains a large enough constant), we observe that $\QQ^*\in\mathcal{U}^*(\tilde\xi)$ and, in turn, that $J_a(\QQ^*,\tilde\xi)\geq R_a$. This verifies the second part of Assumption \ref{ass:1}.b.

\smallskip

Therefore, by Theorem \ref{existence:weak:opt} there exists an optimal contract in the relaxed formulation and it is a limit point of any maximizing sequence of contracts. The latter is an optimal contract in the strong formulation by Theorem \ref{cor:main}.

\section{Appendix}

\emph{Proof of Lemma \ref{lemma:condexpM}.}
We introduce the $(\Omega\times\RR,\mathbb{F},\QQ)$-martingale 
\begin{equation*}
    \phi_t = \EE^\QQ(M|\mathcal{F}_t),\quad t\in[0,T].
\end{equation*}
As $\phi$ is adapted to the filtration of $\mathbb{F}$, w.l.o.g. we can restrict $\phi$ and all processes appearing further in this proof to $(\Omega,\mathbb{F},\WW)$. Then, since $\phi$ is adapted to the completed filtration of a Brownian motion, we know that it has a continuous modification (\cite{LeGall}, Theorem 5.18, consequence 2).
By the martingale representation theorem (\cite{KaratzasShreve}, Theorem 4.2), there exists a progressively measurable process $Z \in \mathbb{H}^2_{loc}([0,T])$ such that
\begin{equation}\label{mtg.representation}
    \phi_t = 1+ \int_0^tZ_sdX_s.
\end{equation} 
We introduce the process 
$$\nu_{t} :=\frac{Z_{t}}{\phi_{t}}, \quad t \in [0,T].$$
We will show that $\nu \in \mathbb{H}^2_{loc}([0,T])$. 
Firstly, we see that $\nu$ is progressive measurable as $\phi$ and $Z$ are progressive measurable. 

Let $\gamma_n$ be the sequence of stopping times defined by:  
\begin{equation*}
    \gamma_n =\inf\{t \geq 0 : \phi_{t} \leq \frac{1}{n}\}\wedge \rho_n,
\end{equation*}
where $\rho_n$ is the localizing sequence of stopping times of $Z$. 

As $\phi$ is continuous we have that 
\begin{equation*}
    \gamma_n \leq \gamma_{n+1}, \quad \WW -a.s.,
\end{equation*}
and
\begin{equation*}
\lim_{n\rightarrow \infty}\gamma_n = \infty.
\end{equation*}
In addition,
\begin{equation*}
    \EE^{ \WW}\left[ \int_0^{T\wedge\gamma_n}\|\nu_{t}\|^2dt\right] \leq n^2 \EE^{\WW}\left[ \int_0^{T\wedge \gamma_n} \|Z_t\|^2dt\right] < \infty.
\end{equation*}
Hence, $\nu \in \mathbb{H}^2_{loc}([0,T])$. Rewriting \eqref{mtg.representation}, we obtain 
\begin{equation*}
    d\phi_t = Z_t dX_t = \phi_t\nu_t dX_t, \quad \phi_0 =1, \quad t \in [0,T].
\end{equation*}
Applying Ito Lemma to $\log(\phi_{T})$, we obtain:
\begin{equation*}
    \begin{split}
    \log(\phi_{T}) = \int_0^{T}\frac{1}{\phi_s}d\phi_s-\frac{1}{2}\int_0^{t}(\phi_s)^{-2}d[\phi]_s =\int_0^{T}\nu_sdX_s-\frac{1}{2}\int_0^{T}\|\nu_t\|^2dt.
    \end{split}
\end{equation*}
Hence,  
\begin{equation}\label{phi.equal.m}
\phi_T = \EE^{\QQ}\left(M|\mathcal{F}^X_{T}\right) = \bar{M}_{T}(\nu), \WW-a.s..
\end{equation}
Taking expectations under $\WW$ in \eqref{phi.equal.m} we obtain that $\EE^{\WW}\bar{M}_{T}(\nu) = 1$. Then, $\bar{M}_{.}(\nu)$ is a martingale. 

Additionally, by the tower property we obtain for all $n\in \NN$:
\begin{equation}\label{ineq.h2.local}
    \begin{split}
\EE^{\QQ}\left[M \int_0^{T\wedge\gamma_n} \|\nu_s\|^2ds\right] &= \EE^{\WW}\left[\bar{M}_{T\wedge\gamma_n}(\nu)\int_0^{T\wedge\gamma_n}\|\nu_s\|^2ds\right] \\
&= 2\EE^{\WW}\left[\bar{M}_{T\wedge\gamma_n}(\nu)\log\left(\bar{M}_{T\wedge\gamma_n}(\nu) \right)\right],
    \end{split}
\end{equation}
Using the submartingale property of $\bar{M}_.\log(\bar{M}_.)$ and Jensen's inequality:
\begin{equation}\label{ineq.0}
    \EE^{\WW}\left[\bar{M}_{T\wedge\gamma_n}(\nu)\log\left(\bar{M}_{T\wedge\gamma_n}(\nu) \right)\right]  \leq \EE^{\WW}\left[\bar{M}_{T}(\nu)\log\left(\bar{M}_{T}(\nu) \right)\right] \leq \EE^{\QQ}\left[M \log(M)\right] <  \infty,
\end{equation}
where we used that $\QQ$ satisfies property 3 in Definition \ref{def:U}. \\
By the monotone convergence theorem applied on \eqref{ineq.0}, we obtain:
\begin{equation}\label{ineq.1}
\begin{split}
\EE^{\QQ}\left[M\int_0^{T}\|\nu_s\|^2ds\right] &= 
\lim_{N\rightarrow\infty } \EE^{\QQ}\left[M\int_0^{T\wedge\gamma_n}\|\nu_s\|^2ds\right] \\
&= \lim_{n \rightarrow \infty} 2\EE^{\WW}\left[\bar{M}_{T\wedge\gamma_n}(\nu)\log(\bar{M}_{T\wedge\gamma_n}(\nu)) \right] \\
&\leq 2\EE^{\WW}\left[\bar{M}_{T}(\nu)\log(\bar{M}_{T}(\nu)) \right].  
\end{split} 
\end{equation}
Moreover, using Fatou's lemma: 
\begin{equation}\label{ineq.2}
\begin{split}
\EE^{\QQ} \left[M\int_0^T\|\nu_s\|^2ds \right] = \lim_{n \rightarrow \infty} 2\EE^{\WW}\left[\bar{M}_{T\wedge\gamma_n}(\nu)\log(\bar{M}_{T\wedge\gamma_n}(\nu)) \right] \geq 2\EE^{\WW}\left[\bar{M}_{T}(\nu)\log(\bar{M}_{T}(\nu)) \right].
\end{split}
\end{equation}
Combining equations \eqref{ineq.1} and \eqref{ineq.2} we obtain: 
\begin{equation*}
    \EE^{\QQ} \left[M \int_0^T \|\nu_s\|^2ds\right]= 2\EE^{\WW}\left[\bar{M}_T(\nu) \log ( \bar{M}_T(\nu))\right].
\end{equation*}
\qed

\bibliographystyle{ieeetr} 
\bibliography{strongweak}
\end{document}